\documentclass[review]{elsarticle}

\usepackage{lineno,hyperref}
\usepackage{amsmath}
\usepackage{amssymb}
\usepackage{algorithmic}
\usepackage{algorithm}
\usepackage{caption}
\usepackage{mathtools}
\usepackage{changes}
\usepackage{multirow}
\usepackage{verbatim}
\usepackage{mathrsfs}
\usepackage{graphicx}
\usepackage{amsmath,amsthm,bm,mathrsfs}

\theoremstyle{plain}
\newtheorem{theorem}{Theorem}[section]
\newtheorem{lemma}[theorem]{Lemma}

\newtheorem{proposition}[theorem]{Proposition}

\theoremstyle{definition}

\theoremstyle{remark}
\newtheorem{remark}{Remark}

\modulolinenumbers[5]

\journal{Journal of \LaTeX\ Templates}

\begin{document}

\begin{frontmatter}

\title{$\mathcal{H}_2$ Model Order Reduction: A Relative Error Setting}

\author[mymainaddress]{Umair~Zulfiqar}

\author[mymainaddress,mysecondaryaddress]{Xin~Du\corref{mycorrespondingauthor}}
\cortext[mycorrespondingauthor]{Corresponding author}
\ead{duxin@shu.edu.com}

\author[mymainaddress]{Qiuyan~Song}

\author[ml]{Muwahida~Liaquat}

\author[vs]{Victor~Sreeram}

\address[mymainaddress]{School of Mechatronic Engineering and Automation, Shanghai University, Shanghai, 200444, China}
\address[mysecondaryaddress]{Shanghai Key Laboratory of Power Station Automation Technology, Shanghai University, Shanghai, 200444, China}
\address[ml]{Department of Electrical Engineering, College of Electrical and Mechanical Engineering, National University of Sciences and Technology, Islamabad, 44000, Pakistan}
\address[vs]{Department of Electrical, Electronic, and Computer Engineering, The University of Western Australia, Perth, 6009, Australia}

\begin{abstract}
In dynamical system theory, the process of obtaining a reduced-order approximation of the high-order model is called model order reduction. The closeness of the reduced-order model to the original model is generally gauged by using system norms of additive or relative error system. The relative error is a superior criterion to the additive error in assessing accuracy in many applications like reduced-order controller and filter designs. In this paper, we propose an oblique projection algorithm that minimizes the $\mathcal{H}_2$ norm of the relative error transfer function. The selection of reduction matrices in the algorithm is motivated by the necessary conditions for local optima of the (squared) $\mathcal{H}_2$ norm of the relative error transfer function. Numerical simulation confirms that the proposed algorithm compares well in accuracy with balanced stochastic truncation while avoiding the solution of large-scale Riccati and Lyapunov equations.
\end{abstract}

\begin{keyword}
$\mathcal{H}_2$ norm\sep model order reduction\sep oblique projection\sep optimal\sep relative error
\end{keyword}

\end{frontmatter}


\section{Introduction}
In mathematical modeling of dynamical systems, the behavior of the systems is generally described by partial differential equations, which are converted to ordinary differential equations to obtain a state-space model. If the number of differential equations describing the dynamical system is high, the state-space model becomes computationally difficult to simulate and analyze. Moreover, the design procedures that take this model as an input also become a computational challenge. To overcome this situation, model order reduction (MOR) is used to obtain a reduced-order model (ROM) that can act as a surrogate for the original high-order model. The ROM closely mimics the original model but is described by a handful of ordinary differential equations. Thus, it is cheaper to simulate and analyze, and it can serve as a surrogate for the original model in the design procedures with tolerable approximation error; cf. \cite{schilders2008model,quarteroni2014reduced,benner2011model,benner2017model,chow2013power}.

The approximation accuracy in MOR can be expressed in various forms and can be quantified by various norms, which lead to various classes of MOR. The most common expression for the closeness of the ROM to the original system is the additive error transfer function, which comprises the bulk of the MOR algorithms available in the literature. A less common yet important expression for the closeness of the ROM to the original system is the relative error transfer function. In control engineering and filter design, this error expression arises quite frequently because of its connection with Bode diagram error in the frequency domain, cf. \cite{obinata2012model}.

Let us denote the transfer functions of the original model and the ROM as $H(s)$ and $\bar{H}_r(s)$, respectively. Then three possible ways to express the MOR problem are the following:
\begin{align}
H(s)&=\bar{H}_r(s)+\Delta_{add}(s),\label{p1}\\
H(s)&=\bar{H}_r(s)\big(I+\Delta_{mul}(s)\big),\label{p2}\\
\bar{H}_r(s)&=H(s)\big(I-\Delta_{rel}(s)\big).\label{p3}
\end{align}
In additive error MOR, $||\Delta_{add}(s)||$ is minimized, whereas, in relative error MOR, $||\Delta_{mul}(s)||$ or $||\Delta_{rel}(s)||$ is minimized. Minimizing $||\Delta_{rel}(s)||$ is closely related to minimizing $||\Delta_{mul}(s)||$, and both can generally be achieved simultaneously with the same MOR algorithm; see \cite{zhou1996robust} for more details on this. The problems (\ref{p2}) and (\ref{p3}) can equivalently be represented as the following
\begin{align}
H(s)&=\big(I+\Delta_{mul}(s)\big)\bar{H}_r(s),\nonumber\\
\bar{H}_r(s)&=\big(I-\Delta_{rel}(s)\big)H(s).\nonumber
\end{align} Throughout this paper, the former representation is used to describe the problems (\ref{p2}) and (\ref{p3}).

Let us discuss the main motivation for considering relative error as a reduction criterion. If a high-order plant is reduced aiming to design a reduced-order controller for the reduced-order plant, $||\Delta_{mul}(s)||$ is a more meaningful and theoretically sound criterion for ensuring robust closed-loop stability. Let $K(s)$ be a controller that stabilizes the reduced-order plant $\bar{H}_r(s)$. To ensure that $K(s)$ also stabilizes $H(s)$, $\bar{H}_r(s)$ should be obtained with a MOR method with the following reduction criterion
\begin{align}
\underset{\substack{\bar{H}_r(s)\\\textnormal{order}=r}}{\text{min}}||[I+K(s)\bar{H}_r(s)]^{-1}K(s)\bar{H}_r(s)\Delta_{mul}(s)||,
\end{align} wherein $\Delta_{mul}(s)=\bar{H}_r^{-1}(s)\Delta_{add}(s)$; cf. \cite{obinata2012model,zhou1996robust,wang1992multiplicative,ennth,vidyasagar1985control}. As shown in \cite{ennth}, an equivalent representation of this reduction criterion can be written as
\begin{align}
\underset{\substack{\bar{H}_r(s)\\\textnormal{order}=r}}{\text{min}}||[\Delta_{mul}(s)]\bar{H}_r(s)K(s)[I+\bar{H}_r(s)K(s)]^{-1}||,
\end{align} wherein $\Delta_{mul}(s)=\Delta_{add}(s)\bar{H}_r^{-1}(s)$. Unfortunately, $K(s)$ and $\bar{H}_r(s)$ in this reduction criterion are not known before reducing the high-order plant $H(s)$ and designing the controller for $\bar{H}_r(s)$. For a well-designed control system, it is reasonable to assume that $\bar{H}_r(s)K(s)[I+\bar{H}_r(s)K(s)]^{-1}$ or $[I+K(s)\bar{H}_r(s)]^{-1}K(s)\bar{H}_r(s)$ is approximately $I$ over the operating bandwidth of the system and $0$ outside the bandwidth. This condition can be achieved approximately by a good robust control design for the reduced-order plant. This point will be highlighted further in the numerical examples section. Under this condition, the problem under consideration reduces to the relative error MOR problem. If $||\Delta_{mul}(s)||$ is small within the operating bandwidth, $K(s)$ will provide good closed-loop stability with $H(s)$ since error outside this frequency region is attenuated by the small frequency response of $\bar{H}_r(s)K(s)[I+\bar{H}_r(s)K(s)]^{-1}$ or $[I+K(s)\bar{H}_r(s)]^{-1}K(s)\bar{H}_r(s)$. Since the operating frequency is not always known beforehand in the absence of a controller, $\underset{\substack{\bar{H}_r(s)\\\textnormal{order}=r}}{\text{min}}||\Delta_{mul}(s)||$ is an effective reduction criterion for the problem under consideration. Moreover, the problem of obtaining a reduced-order infinite impulse response (IIR) filter from the high-order finite impulse response (FIR) filter is also equivalent to solving the problem (\ref{p2}) or (\ref{p3}); cf. \cite{obinata2012model}.

Balanced truncation (BT) is a well-known MOR algorithm that ensures that $||\Delta_{add}(s)||_{\mathcal{H}_\infty}$ is small \cite{moore1981principal}. An \textit{apriori} upper bound on $||\Delta_{add}(s)||_{\mathcal{H}_\infty}$ also holds for the ROM \cite{enns1984model}. However, BT is a computationally expensive algorithm, and its applicability is only viable for the order of a few hundred. By replacing the Lyapunov equations with their low-rank approximations, the computational cost of BT can be reduced, and its applicability can be extended to large-scale systems \cite{gugercin2003modified}. In \cite{enns1984model,ennth}, BT is extended to frequency-weighted BT (FWBT) to ensure that $||W(s)\Delta_{add}(s)V(s)||_{\mathcal{H}_\infty}$ is small, wherein $W(s)$ and $V(s)$ are frequency weights to emphasize the frequency region within which superior accuracy is required. FWBT can be used to solve problems (\ref{p2}) and (\ref{p3}) when $H(s)$ is a stable minimum phase system \cite{zhou1995frequency}. \textit{Apriori} upper bounds on $||\Delta_{mul}(s)||_{\mathcal{H}_\infty}$ and $||\Delta_{rel}(s)||_{\mathcal{H}_\infty}$ also hold in this scenario \cite{zhou1995frequency}. When $H(s)$ is not a minimum phase system, problems (\ref{p2}) and (\ref{p3}) can be solved using a variant of BT called balanced stochastic truncation (BST) \cite{green1988balanced}. Similar \textit{apriori} upper bounds on $||\Delta_{mul}(s)||_{\mathcal{H}_\infty}$ and $||\Delta_{rel}(s)||_{\mathcal{H}_\infty}$ hold in this scenario as well \cite{green1988relative}. Moreover, BST guarantees that $\bar{H}_r(s)$ is a minimum phase system. BST is even more computationally expensive than BT. The Riccati and Lyapunov equations in BST can be replaced with their low-rank approximations to extend its applicability to large-scale systems; see \cite{benner2001efficient}. BST is generalized for the frequency-weighted scenario in \cite{kim1995multiplicative} to minimize $||\Delta_{rel}(s)V(s)||_{\mathcal{H}_\infty}$ for obtaining a reduced-order controller. Some other generalizations of BST are also reported, like \cite{shaker2008frequency} and \cite{tahavori2013model}.

Several performance criteria in control theory and signal processing are quantified by the $\mathcal{H}_\infty$ norm. However, in a large-scale setting, the computation of the $\mathcal{H}_\infty$ norm is expensive, and generally, the MOR algorithms seeking to reduce the $\mathcal{H}_\infty$ norm of the error system are also expensive \cite{castagnotto2018optimal}. From a computational standpoint, the $\mathcal{H}_2$ norm is a better option as it is easy to compute due to its relationship with system gramians, for which a low-rank solution is generally available. The optimal solution in the $\mathcal{H}_2$ norm, i.e., local optimum of $||\Delta_{add}(s)||_{\mathcal{H}_2}^2$, can be obtained efficiently even in a large-scale setting \cite{xu2011optimal,bunse2010h2}. In contrast, the optimal solution in the $\mathcal{H}_\infty$ norm is computationally expensive, even in a small-scale setting \cite{castagnotto2017interpolatory}. From a system theory perspective, reducing the $\mathcal{H}_2$ or $\mathcal{H}_\infty$ norm of $\Delta_{add}(s)$ essentially means ensuring that the output response of $\Delta_{add}(s)$ is small for any input. Therefore, the MOR algorithms that are optimized to reduce $||\Delta_{add}(s)||_{\mathcal{H}_2}$ inevitably also reduce $||\Delta_{add}(s)||_{\mathcal{H}_\infty}$. This makes the $\mathcal{H}_2$ norm an attractive option for performance specification in MOR due to its favorable numerical properties \cite{castagnotto2017interpolatory}.

In \cite{gugercin2008h_2}, an iterative Rational Krylov algorithm (IRKA) is presented for single-input single-output (SISO) systems to construct a local optimum of $||\Delta_{add}$ $(s)||_{\mathcal{H}_2}^2$. IRKA is extended to MIMO systems in \cite{van2008h2}. A more robust and efficient approach for constructing a local optimum of $||\Delta_{add}(s)||_{\mathcal{H}_2}^2$ using Sylvester equations is presented in \cite{xu2011optimal,benner2011sparse}. The convergence in IRKA is not guaranteed. Some trust region-based approaches to speed up convergence are reported in the literature like \cite{beattie2009trust,sato2015riemannian}. Some suboptimal solutions to the $\mathcal{H}_2$-optimal MOR problem are also reported in \cite{wolf2013,ibrir2018projection}, which preserves some other system properties like stability. The algorithm presented in \cite{benner2011sparse} is generalized for the frequency weighted scenario in \cite{zulfiqar2021frequency}, which constructs a near-optimum of $||W(s)\Delta_{add}(s)V(s)||_{\mathcal{H}_2}^2$. Similarly, IRKA is generalized for the frequency weighted scenario in \cite{zulfiqar2021}, which also constructs a near-optimum of $||W(s)\Delta_{add}(s)V(s)||_{\mathcal{H}_2}^2$. To the best of the author's knowledge, there is no MOR algorithm available in the literature that seeks to minimize $||\Delta_{mul}(s)||_{\mathcal{H}_2}$ and $||\Delta_{rel}(s)||_{\mathcal{H}_2}$, which has motivated the results of this paper.

In this paper, it is first demonstrated that the iterative frequency-weighted $\mathcal{H}_2$ MOR algorithm (IFWHMORA) proposed in \cite{zulfiqar2021frequency} can be used for solving the problem (\ref{p3}), and a near-optimum of $||\Delta_{rel}(s)||_{\mathcal{H}_2}^2$ can be constructed when $H(s)$ is a stable minimum phase system. The case when $H(s)$ is not a minimum phase system is also discussed. Then necessary conditions for the local optimum of $||\Delta_{mul}(s)||_{\mathcal{H}_2}^2$ are derived when $\bar{H}_r(s)$ is a minimum phase system. The case when $\bar{H}_r(s)$ is not a minimum phase system is also discussed, and an oblique projection algorithm, inspired by the necessary conditions of the local optimum, is proposed. The proposed algorithm is tested on three benchmark dynamical systems. The numerical simulation confirms that the proposed algorithm compares well in accuracy with BST while avoiding the solution of large-scale Lyapunov and Riccati equations.
\section{Preliminaries}
In this section, the problem statement is formulated, and two important MOR techniques in the existing literature for this problem are briefly reviewed.
\subsection{Problem Setting}
Consider an $n^{th}$-order $m\times m$ square linear time-invariant (LTI) system $H(s)$. Let $(A,B,C,D)$ be its state-space realization, which is related to $H(s)$ as
\begin{align}
H(s)=C(sI-A)^{-1}B+D\nonumber
\end{align} where $A\in\mathbb{R}^{n\times n}$, $B\in\mathbb{R}^{n\times m}$, $C\in\mathbb{R}^{m\times n}$, and $D\in\mathbb{R}^{m\times m}$.

Let us denote the $r^{th}$-order ($r\ll n$) approximation of $H(s)$ as $\bar{H}_r(s)$, which is also an $m\times m$ square linear time-invariant (LTI) system. Let $(\bar{A}_r,\bar{B}_r,\bar{C}_r,D)$ be its state-space realization, which is related to $\bar{H}_r(s)$ as
\begin{align}
\bar{H}_r(s)=\bar{C}_r(sI-\bar{A}_r)^{-1}\bar{B}_r+D\nonumber
\end{align} where $\bar{A}_r\in\mathbb{R}^{r\times r}$, $\bar{B}_r\in\mathbb{R}^{r\times m}$, and $\bar{C}_r\in\mathbb{R}^{m\times n}$.

Let the matrices $\bar{V}_r\in\mathbb{R}^{n\times r}$ and $\bar{W}_r\in\mathbb{R}^{n\times r}$ be such that $\bar{W}_r^T\bar{V}_r=I$, and the columns of $\bar{V}_r$ span an $r$-dimensional subspace along with the kernels of $\bar{W}_r^T$. If $\bar{H}_r(s)$ is constructed via oblique projection, the following holds
\begin{align}
\bar{A}_r&=\bar{W}_r^TA\bar{V}_r,&\bar{B}_r&=\bar{W}_r^TB,&\bar{C}_r&=C\bar{V}_r\nonumber
\end{align} where $\Pi=\bar{V}_r\bar{W}_r^T$ is an oblique projector.

The $\mathcal{H}_2$ relative-error MOR is to construct $\bar{H}_r(s)$, such that $||\Delta_{rel}(s)||_{\mathcal{H}_2}$ or $||\Delta_{mul}(s)||_{\mathcal{H}_2}$ is small, i.e.,
\begin{align}
\underset{\substack{\bar{H}_r(s)\\\textnormal{order}=r}}{\text{min}}||\Delta_{rel}(s)||_{\mathcal{H}_2}&&\textnormal{and/or}&&\underset{\substack{\bar{H}_r(s)\\\textnormal{order}=r}}{\text{min}}||\Delta_{mul}(s)||_{\mathcal{H}_2}.\nonumber
\end{align}If $\bar{H}_r(s)$ is obtained via oblique projection, the problem becomes that of constructing the reduction matrices $\bar{V}_r$ and $\bar{W}_r$, such that $||\Delta_{rel}(s)||_{\mathcal{H}_2}$ or $||\Delta_{mul}(s)||_{\mathcal{H}_2}$ is small. For clarity, $||\Delta_{rel}(s)||_{\mathcal{H}_2}$ and $||\Delta_{mul}(s)||_{\mathcal{H}_2}$ are referred to as relative error I and relative error II, respectively, in the remainder of the paper.
\subsection{Literature Review}
In this subsection, two important algorithms for relative-error MOR in the literature are reviewed. The first one minimizes $||\Delta_{rel}(s)||_{\mathcal{H}_\infty}$ and $||\Delta_{mul}(s)||_{\mathcal{H}_\infty}$, while the second one can minimize $||\Delta_{rel}(s)||_{\mathcal{H}_2}$.
\subsubsection{Balanced Stochastic Truncation (BST)}
Let $G(s)$ be a minimum phase right spectral factor of $\Phi(s)=H(s)H^*(s)$. The controllability gramian $P_{11}$ of the pair $(A,B)$ solves the following Lyapunov equation
\begin{align}
AP_{11}+P_{11}A^T+BB^T=0.\label{nnee6}
\end{align} Let us define $B_x$ and $A_x$ as the following
\begin{align}
B_x&=P_{11}C^T+BD^T,&A_x&=A-B_x(DD^T)^{-1}C.\nonumber
\end{align}
Now, let $X_\phi$ be the solution of the following Riccati equation
\begin{align}
A_x^TX_\phi+X_\phi A_x+X_\phi B_x(DD^T)^{-1}B_x^TX_\phi+C^T(DD^T)^{-1}C=0.\label{nnee6666}
\end{align}
The reduction matrices in BST are computed such that $\bar{W}_r^TP_{11}\bar{W}_r=\bar{V}_r^TX_\phi\bar{V}_r=diag(\sigma_1,\cdots,\sigma_r)$ where $\sigma_i$ are $r$ biggest Hankel singular values of $G^{-*}(s)H(s)$.
\begin{remark}
When $H(s)$ is a minimum phase system, and the matrix $D$ has full rank, a stable state-space realization of $H^{-1}(s)$ exists, and BST reduces to FWBT by selecting the frequency weight as $W(s)=H^{-1}(s)$; see \cite{zhou1995frequency}.
\end{remark}
\subsubsection{Iterative Frequency-weighted $\mathcal{H}_2$ MOR Algorithm (IFWHMORA)}
Let $(A_i,B_i,C_i,D_i)$ be a stable state-space realization of the frequency weight $W(s)$. Let $P_{12}$, $Q_{12}$, $Q_{13}$, $Q_{23}$, and $Q_{i}$ satisfy the following linear matrix equations
\begin{align}
AP_{12}+P_{12}\bar{A}_r^T+B\bar{B}_r^T&=0,\hspace*{2cm}\label{nnee7}\\
A^TQ_{12}+Q_{12}\bar{A}_r-\bar{C}_r^TB_i^TQ_{23}^T-Q_{13}B_i\bar{C}_r-C^TD_i^TD_i\bar{C}_r&=0,\label{supn1}\\
A^TQ_{13}+Q_{13}A_i+C^T(B_i^TQ_i+D_i^TC_i)&=0,\label{supn2}\\
\bar{A}_r^TQ_{23}+Q_{23}A_i-\bar{C}_r^T(B_i^TQ_i+D_i^TC_i)&=0,\label{supn3}\\
A_i^TQ_i+Q_iA_i+C_i^TC_i&=0.\label{supn4}
\end{align}
In IFWHMORA, starting with an arbitrary guess of the ROM, the reduction matrices are selected as $\bar{V}_r=P_{12}$ and $\bar{W}_r=Q_{12}$. The oblique projection condition $\bar{W}_r^T\bar{V}_r=I$ is enforced by using bi-orthogonal Gram-Schmidt method. The process is repeated until convergence upon which a near-optimum of $||W(s)\Delta_{add}(s)||_{\mathcal{H}_2}^2$ is obtained.
\begin{remark}
Although not noted in \cite{zulfiqar2021frequency}, IFWHMORA can be used to obtain a near optimum of $||\Delta_{rel}(s)||_{\mathcal{H}_2}^2$, which will be shown in the next section.
\end{remark}
\section{MOR based on Relative Error Criterion I}
Let us assume that $H(s)$ is a minimum phase system, and the matrix $D$ is full rank. Then a possible state-space realization $(A_i,B_i,C_i,D_i)$ for $W(s)=H^{-1}(s)$ is given as
\begin{align}
A_i&=A-BD^{-1}C,& B_i&=-BD^{-1}, & C_i&=D^{-1}C,& D_i&=D^{-1},\label{neq4}
\end{align}cf. \cite{zhou1996robust}. Further, it can readily be noted that $\Delta_{rel}(s)$ becomes equal to $H^{-1}(s)\Delta_{add}(s)$. Then, by using IFWHMORA, a near-optimum of $H^{-1}(s)\Delta_{add}(s)$ can be obtained.

The state-space realization of $H^{-1}(s)$ given in (\ref{neq4}) only exists when $D$ is invertible. Thus, when $D$ is rank deficient, IFWHMORA cannot be used without mathematical manipulation. Moreover, when $H(s)$ has zeros in the right half of the $s$-plane, i.e., $H(s)$ is not a minimum phase system, $H^{-1}(s)$ becomes unstable, and $||\Delta_{rel}(s)||_{\mathcal{H}_2}$ becomes unbounded. This also limits the applicability of IFWHMORA. We now discuss remedies for these two issues in the sequel.

Let $G(s)$ be a minimum phase right spectral factor of $H(s)H^*(s)$, such that $G^*(s)G(s)=H(s)H^*(s)$. From the definition of the $\mathcal{H}_2$ norm, $||\Delta_{rel}(s)||_{\mathcal{H}_2}$ can be defined as
  \begin{align}
  ||\Delta_{rel}(s)||_{\mathcal{H}_2}
  &=\sqrt{\frac{1}{2\pi}\int_{-\infty}^{\infty}trace\Big(\Delta_{add}^*(j\omega)[H(j\omega)H^{*}(j\omega)]^{-1}\Delta_{add}(j\omega)\Big)d\omega}\nonumber\\
  &=\sqrt{\frac{1}{2\pi}\int_{-\infty}^{\infty}trace\Big(\Delta_{add}^*(j\omega)G^{-1}(j\omega)G^{-*}(j\omega)\Delta_{add}(j\omega)\Big)d\omega}.\label{neq5}
  \end{align}
It can readily be noted from (\ref{neq5}) that by selecting the frequency weight as $W(s)=G^{-*}(s)$, IFWHMORA can be used even if $H(s)$ is not a minimum phase system. Next, we formulate a possible state-space realization of $G^{-*}(s)$ on similar lines to BST. A possible realization of $G(s)$ is computed in BST as
\begin{align}
A_{sp}&=A,&B_{sp}&=B_x,&C_{sp}&=D^{-1}(C-B_x^TX_{\phi}),&D_{sp}&=D^T.\nonumber
\end{align}
  Then a possible minimum phase realization of $G^*(s)$ is given as
  \begin{align}
  A_g&=-A_{sp}^T,&B_g&=-C_{sp}^T,&C_g&=B_{sp}^T,&D_g=D_{sp}^T,\nonumber
  \end{align} cf. \cite{zhou1996robust}. Further, a possible stable realization of $G^{-*}(s)$ is given as
  \begin{align}
A_{g_i}&=A_g-B_gD_g^{-1}C_g,& B_{g_i}&=-B_gD_g^{-1}, & C_{g_i}&=D_g^{-1}C_g,& D_{g_i}&=D_g^{-1}.\nonumber
\end{align}
If the matrix $D$ is not full rank, its invertibility becomes an issue. This can be tackled by the so-called $\epsilon$-regularization, i.e., by replacing the original $D$ matrix with an artificial $D=\epsilon I$ in the algorithm wherein $\epsilon$ is a small valued scalar. In the final ROM, the $D$-matrix is set to the original $D$-matrix. This approach is an effective one, and it has been used in BST to deal with a similar issue; for details, see \cite{green1988balanced,benner2001efficient} and also the documentation of MATLAB's built-in function for BST \cite{bstmr}.

A major drawback of using IFWHMORA is that the equations (\ref{supn1})-(\ref{supn4}), (\ref{nnee6}), and (\ref{nnee6666}) are expensive to solve in a large-scale setting. The number of large-scale linear matrix equations in IFWMORA is even more than that in BST. The main motivation for developing an algorithm, which tends to minimize $||\Delta_{rel}(s)||_{\mathcal{H}_2}$ instead of $||\Delta_{rel}(s)||_{\mathcal{H}_\infty}$, was to avoid large-scale Riccati and Lyapunov equations involved in BST. Even when $H(s)$ is a minimum phase system, the frequency weight $W(s)=H^{-1}(s)$ remains a large-scale system, which seriously affects the computational efficiency of IFWHMORA because the equations (\ref{supn1})-(\ref{supn4}) remain expensive to solve. To conclude, IFWHMORA is not a computationally viable option for constructing a near-optimum (local) of $||\Delta_{rel}(s)||_{\mathcal{H}_2}^2$ in a large-scale setting.
\section{MOR based on Relative Error Criterion II}
If $\bar{H}_r(s)$ is invertible, $\Delta_{mul}(s)$ can be expressed as $\bar{H}_r^{-1}(s)\Delta_{add}(s)$. Thus, the relative error II can also be seen as a frequency-weighted additive error criterion. The theory developed for IFWHMORA in \cite{zulfiqar2021frequency} assumed that $W(s)$ does not depend on $\bar{H}_r(s)$. Here, however, the frequency-weight depends on the ROM, which is not known \textit{apriori}. Thus, IFWHMORA cannot be used to obtain a near-optimum of $||\Delta_{mul}(s)||_{\mathcal{H}_2}^2$. In the sequel, we derive the necessary conditions for the local optimum of $||\Delta_{mul}(s)||_{\mathcal{H}_2}^2$ and note that these are different from the conditions derived in \cite{petersson2013nonlinear} for $||W(s)\Delta_{add}(s)||_{\mathcal{H}_2}^2$ when $W(s)$ is replaced by $\bar{H}^{-1}_r(s)$.
\subsection{State-space Realization of $\Delta_{mul}(s)$ and $||\Delta_{mul}(s)||_{\mathcal{H}_2}$}
When $\bar{H}_r(s)$ is a stable minimum phase, and the matrix $D$ is full rank, a possible state-space realization of $\Delta_{mul}(s)=\bar{H}_r^{-1}(s)\Delta_{add}(s)$ is given as
\begin{align}
A_{mul}&=\begin{bmatrix}A&0&0\\0&\bar{A}_r&0\\-\bar{B}_rD^{-1}C& \bar{B}_rD^{-1}\bar{C}_r&\bar{A}_r-\bar{B}_rD^{-1}\bar{C}_r\end{bmatrix},& B_{mul}&=\begin{bmatrix}B\\\bar{B}_r\\0\end{bmatrix},\nonumber\\
C_{mul}&=\begin{bmatrix}D^{-1}C&-D^{-1}\bar{C}_r&D^{-1}\bar{C}_r\end{bmatrix};\nonumber
\end{align}cf. \cite{zhou1996robust}.
Let $P_{mul}=\begin{bmatrix}P_{11}&P_{12}&\hat{P}_{13}\\P_{12}^T&\hat{P}_{22}&\hat{P}_{23}\\\hat{P}_{13}^T&\hat{P}_{23}^T&\hat{P}_{33}\end{bmatrix}$ be the controllability gramian of the pair $(A_{mul},B_{mul})$, which solves the following Lyapunov equation
\begin{align}
A_{mul}P_{mul}+P_{mul}A_{mul}^T+B_{mul}B_{mul}^T&=0\nonumber
\end{align} wherein
\begin{align}
A\hat{P}_{13}+\hat{P}_{13}(\bar{A}_r-\bar{B}_rD^{-1}\bar{C}_r)^T-P_{11}C^TD^{-T}\bar{B}_r^T+P_{12}\bar{C}_r^TD^{-T}\bar{B}_r^T&=0,\label{v2.21}\\
\bar{A}_r\hat{P}_{22}+\hat{P}_{22}\bar{A}_r^T+\bar{B}_r\bar{B}_r^T&=0,\label{nnee8}\\
\bar{A}_r\hat{P}_{23}+\hat{P}_{23}(\bar{A}_r-\bar{B}_rD^{-1}\bar{C}_r)^T-P_{12}^TC^TD^{-T}\bar{B}_r^T+P_{22}\bar{C}_r^TD^{-T}\bar{B}_r^T&=0,\label{v2.22}\\
(\bar{A}_r-\bar{B}_rD^{-1}\bar{C}_r)\hat{P}_{33}+\hat{P}_{33}(\bar{A}_r-\bar{B}_rD^{-1}\bar{C}_r)^T-\bar{B}_rD^{-1}C\hat{P}_{13}&\nonumber\\
+\bar{B}_rD^{-1}\bar{C}_r\hat{P}_{23}-\hat{P}_{13}^TC^TD^{-T}\bar{B}_r^T+\hat{P}_{23}^T\bar{C}_r^TD^{-T}\bar{B}_r^T&=0.\label{v2.23}
\end{align}
Let $Q_{mul}=\begin{bmatrix}\hat{Q}_{11}&\hat{Q}_{12}&\hat{Q}_{13}\\\hat{Q}_{12}^T&\hat{Q}_{22}&\hat{Q}_{23}\\\hat{Q}_{13}^T&\hat{Q}_{23}&\hat{Q}_{33}\end{bmatrix}$ be the observability gramian of the pair $(A_{mul},C_{mul})$, which solves the following Lyapunov equation
\begin{align}
A_{mul}^TQ_{mul}+Q_{mul}A_{mul}+C_{mul}^TC_{mul}&=0\nonumber
\end{align} wherein
\begin{align}
A^T\hat{Q}_{11}+\hat{Q}_{11}A-C^TD^{-T}\bar{B}_r^T\hat{Q}_{13}^T-\hat{Q}_{13}\bar{B}_rD^{-1}C+C^TD^{-T}D^{-1}C&=0,\label{v2.25}\\
A^T\hat{Q}_{12}+\hat{Q}_{12}\bar{A}_r-C^TD^{-T}\bar{B}_r^T\hat{Q}_{23}+\hat{Q}_{13}\bar{B}_rD^{-1}\bar{C}_r-C^TD^{-T}D^{-1}\bar{C}_r&=0,\label{v2.26}\\
A^T\hat{Q}_{13}+\hat{Q}_{13}(\bar{A}_r-\bar{B}_rD^{-1}\bar{C}_r)-C^TD^{-T}\bar{B}_r^T\hat{Q}_{33}+C^TD^{-T}D^{-1}\bar{C}_r&=0,\label{v2.27}\\
\bar{A}_r^T\hat{Q}_{22}+\hat{Q}_{22}\bar{A}_r+\bar{C}_r^TD^{-T}\bar{B}_r^T\hat{Q}_{23}+\hat{Q}_{23}\bar{B}_rD^{-1}\bar{C}_r+\bar{C}_r^TD^{-T}D^{-1}\bar{C}_r&=0,\label{v2.28}\\
\bar{A}_r^T\hat{Q}_{23}+\hat{Q}_{23}(\bar{A}_r-\bar{B}_rD^{-1}\bar{C}_r)+\bar{C}_r^TD^{-T}\bar{B}_r^T\hat{Q}_{33}-\bar{C}_r^TD^{-T}D^{-1}\bar{C}_r&=0,\label{v2.29}\\
(\bar{A}_r-\bar{B}_rD^{-1}\bar{C}_r)^T\hat{Q}_{33}+\hat{Q}_{33}(\bar{A}_r-\bar{B}_rD^{-1}\bar{C}_r)+\bar{C}_r^TD^{-T}D^{-1}\bar{C}_r&=0.\label{v2.30}
\end{align}
\begin{proposition}
The matrices $\hat{Q}_{12}$ and $\hat{Q}_{22}$ are equal to $-\hat{Q}_{13}$ and $-\hat{Q}_{23}$, respectively. Further $\hat{Q}_{33}$ is equal to $\hat{Q}_{22}$.
\end{proposition}
\begin{proof}
By adding equations (\ref{v2.29}) and (\ref{v2.30}), we get
\begin{align}
(\bar{A}_r-\bar{B}_rD^{-1}\bar{C}_r)^T(\hat{Q}_{23}+\hat{Q}_{33})+(\hat{Q}_{23}+\hat{Q}_{33})(\bar{A}_r-\bar{B}_rD^{-1}\bar{C}_r)&=0.\nonumber
\end{align}
Thus $\hat{Q}_{23}+\hat{Q}_{33}=0$ and $\hat{Q}_{33}=-\hat{Q}_{23}$. Similarly, by subtracting equation (\ref{v2.28}) from (\ref{v2.30}), we get
\begin{align}
(\bar{A}_r-\bar{B}_rD^{-1}\bar{C}_r)^T(\hat{Q}_{33}-\hat{Q}_{22})+(\hat{Q}_{33}-\hat{Q}_{22})(\bar{A}_r-\bar{B}_rD^{-1}\bar{C}_r)&=0.\nonumber
\end{align}
Thus $\hat{Q}_{33}-\hat{Q}_{22}=0$ and $\hat{Q}_{33}=\hat{Q}_{22}$. Finally, by adding equation (\ref{v2.26}) and (\ref{v2.27}), we get
\begin{align}
\bar{A}^T(\hat{Q}_{12}+\hat{Q}_{13})+(\hat{Q}_{12}+\hat{Q}_{13})\bar{A}_r&=0.\nonumber
\end{align}
Thus $\hat{Q}_{12}+\hat{Q}_{13}=0$ and $\hat{Q}_{12}=-\hat{Q}_{13}$.
\end{proof}
From the definition of $\mathcal{H}_2$ norm, $||\Delta_{mul}(s)||_{\mathcal{H}_2}$ is given as
\begin{align}
||\Delta_{mul}(s)||_{\mathcal{H}_2}&=\sqrt{\frac{1}{2\pi}\int_{-\infty}^{\infty}trace\Big(\Delta_{mul}^*(j\omega)\Delta_{mul}(j\omega)\Big)d\omega}\nonumber\\
  &=\sqrt{\frac{1}{2\pi}\int_{-\infty}^{\infty}trace\Big(\Delta_{add}^*(j\omega)\bar{H}_r^{-*}(j\omega)\bar{H}_r^{-1}(j\omega)\Delta_{add}(j\omega)\Big)d\omega}\nonumber\\
  &=\sqrt{trace(C_{mul}P_{mul}C_{mul}^T)}\nonumber\\
&=\big(trace(D^{-1}CP_{11}C^TD^{-T}-2D^{-1}CP_{12}\bar{C}_r^TD^{-T}\nonumber\\
&\hspace*{3cm}+2D^{-1}C\hat{P}_{13}\bar{C}_r^TD^{-T}+D^{-1}\bar{C}_r\hat{P}_{22}\bar{C}_r^TD^{-T}\nonumber\\
&\hspace*{3cm}-2D^{-1}\bar{C}_r\hat{P}_{23}\bar{C}_r^TD^{-T}+D^{-1}\bar{C}_r\hat{P}_{33}\bar{C}_r^TD^{-T})\big)^{\frac{1}{2}}\nonumber\\
&=\sqrt{trace(B_{mul}^TQ_{mul}B_{mul})}\nonumber\\
&=\big(trace(B^T\hat{Q}_{11}B+2B^T\hat{Q}_{12}\bar{B}_r+\bar{B}_r^T\hat{Q}_{22}\bar{B}_r)\big)^{\frac{1}{2}}.\nonumber
\end{align}
\subsection{Necessary Conditions for the Local Optimum of $||\Delta_{mul}(s)||_{\mathcal{H}_2}^2$}
In this subsection, the necessary conditions for local optimum $||\Delta_{mul}(s)||_{\mathcal{H}_2}^2$ are stated in terms of the state-space realization and the $\mathcal{H}_2$ norm expression derived in the previous subsection. The following three properties of trace \cite{petersen2008matrix} are used in the derivation of the necessary conditions:
\begin{enumerate}
  \item Additive property, i.e., $trace(R+S+T)=trace(R)+trace(S)+trace(T)$.
  \item Transpose property, i.e., $trace(R)=trace(R^T)$.
  \item Cyclic permutation property, i.e., $trace(RST)=trace(TRS)=trace(STR)$.
\end{enumerate}
\begin{lemma}\label{lemma}
Let $P$ and $Q$ satisfy the following Sylvester equations
\begin{align}
\hat{A}P+P\hat{B}+\hat{C}&=0,\nonumber\\
\hat{B}Q+Q\hat{A}+\hat{D}&=0.\nonumber
\end{align}
Then $trace(\hat{C}Q)=trace(\hat{D}P)$.
\end{lemma}
\begin{proof}
See the proof of Lemma 4.1 in \cite{petersson2013nonlinear}.
\end{proof}
\begin{theorem}\label{th1}
The state-space realization of the local optimum of the $||\Delta_{mul}(s)||_{\mathcal{H}_2}^2$ satisfies the following conditions
\begin{align}
\hat{Q}_{12}^TP_{12}+\hat{Q}_{22}\hat{P}_{22}+\bar{X}&=0,\label{opc1}\hspace*{2cm}\\
\hat{Q}_{12}^TB+\hat{Q}_{22}\bar{B}_r+\bar{Y}&=0,\label{opc2}\\
      -(D^{T}D)^{-1}CP_{12}+(D^{T}D)^{-1}\bar{C}_r\hat{P}_{22}+\bar{Z}&=0\label{opc3}
\end{align}
wherein
\begin{align}
\bar{X}&=\hat{Q}_{13}^T\hat{P}_{13}+\hat{Q}_{23}\hat{P}_{23}^T+\hat{Q}_{23}\hat{P}_{23}+\hat{Q}_{33}\hat{P}_{33},\nonumber\\
\bar{Y}&=\big(-\hat{Q}_{13}^TP_{11}C^T+\hat{Q}_{13}^TP_{12}\bar{C}_r^T-\hat{Q}_{23}P_{12}^TC^T\nonumber\\
&\hspace*{2.9cm}+\hat{Q}_{23}\hat{P}_{22}\bar{C}_r^T-\hat{Q}_{13}^T\hat{P}_{13}\bar{C}_r^T-\hat{Q}_{33}\hat{P}_{13}^TC^T\nonumber\\
&\hspace*{2.9cm}-\hat{Q}_{23}\hat{P}_{23}\bar{C}_r^T+\hat{Q}_{33}\hat{P}_{23}^T\bar{C}_r^T-\hat{Q}_{33}\hat{P}_{33}\bar{C}_r^T\big)D^{-T},\nonumber\\
\bar{Z}&=D^{-T}\big(D^{-1}C\hat{P}_{13}-D^{-1}\bar{C}_r\hat{P}_{23}-D^{-1}\bar{C}_r\hat{P}_{23}^T+D^{-1}\bar{C}_r\hat{P}_{33}\nonumber\\
&\hspace*{3.1cm}+\bar{B}_r^T\hat{Q}_{13}^T\hat{P}_{13}+\bar{B}_r^T\hat{Q}_{13}^TP_{12}+\bar{B}_r^T\hat{Q}_{23}\hat{P}_{23}\nonumber\\
&\hspace*{3.1cm}+\bar{B}_r^T\hat{Q}_{23}\hat{P}_{22}+\bar{B}_r^T\hat{Q}_{33}\hat{P}_{33}+\bar{B}_r^T\hat{Q}_{33}\hat{P}_{23}^T\big).\nonumber
\end{align}
\end{theorem}
\begin{proof}
The proof is given in the appendix.
    \end{proof}
\begin{remark}
In \cite{petersson2013nonlinear}, the necessary conditions for the local optimum of $||W(s)\Delta_{add}$ $(s)V(s)||_{\mathcal{H}_2}^2$ are derived. It can be easily verified that by putting $W(s)=\bar{H}^{-1}_r(s)$ and $V(s)=I$, the necessary conditions in Theorem \ref{th1} cannot be obtained. This is because it is assumed in \cite{petersson2013nonlinear} that $W(s)$ does not depend on $\bar{H}_r(s)$.
\end{remark}
\subsection{An Oblique Projection Algorithm}
Let us assume that $\hat{Q}_{22}$ and $\hat{P}_{22}$ are invertible. Then according to the optimality conditions (\ref{opc2}) and (\ref{opc3}), the optimal choices for $\bar{B}_r$ and $\bar{C}_r$ are the following
\begin{align}
\bar{B}_r&=-\hat{Q}_{22}^{-1}\hat{Q}_{12}^TB-\hat{Q}_{22}^{-1}\bar{Y}&\bar{C}_r&=CP_{12}\hat{P}_{22}^{-1}-(D^TD)\bar{Z}\hat{P}_{22}^{-1}.\nonumber
\end{align}
If the ROM is constructed via oblique projection, this suggests selecting the reduction matrices as $\bar{W}_r=-\hat{Q}_{12}\hat{Q}_{22}^{-1}$ and $\bar{V}_r=P_{12}\hat{P}_{22}^{-1}$. Then, by the virtue of oblique projection condition, $\bar{W}_r^T\bar{V}_r=I$, $\hat{Q}_{12}^TP_{12}+\hat{Q}_{22}\hat{P}_{22}=0$. Further, it can readily be noted from the equation (\ref{v2.22}) that $P_{23}=0$ when $\bar{C}_r=CP_{12}\hat{P}_{22}^{-1}$. The deviations in the satisfaction of optimality conditions (\ref{opc1})-(\ref{opc3}) with this selection of reduction matrices in the oblique projection framework are respectively given as
\begin{align}
\hspace*{1cm}d_1&=\hat{Q}_{13}^T\hat{P}_{13}+\hat{Q}_{33}\hat{P}_{33},\nonumber\\
d_2&=\big(-\hat{Q}_{13}^TP_{11}+\hat{Q}_{13}^TP_{12}\hat{P}_{22}^{-1}P_{12}^T-\hat{Q}_{13}^T\hat{P}_{13}\hat{P}_{22}^{-1}P_{12}^T\nonumber\\
&\hspace*{2.35cm}-\hat{Q}_{33}\hat{P}_{13}^T-\hat{Q}_{33}\hat{P}_{33}\hat{P}_{22}^{-1}P_{12}^T\big)C^TD^{-T},\nonumber\\
d_3&=D^{-T}\big(D^{-1}C\hat{P}_{13}+D^{-1}CP_{12}\hat{P}_{22}^{-1}\hat{P}_{33}+B^T\hat{Q}_{13}\hat{Q}_{22}^{-1}\hat{Q}_{13}^T\hat{P}_{13}\nonumber\\
&\hspace*{3.1cm}+B^T\hat{Q}_{13}\hat{Q}_{22}^{-1}\hat{Q}_{13}^TP_{12}+B\hat{Q}_{13}\hat{P}_{22}+B^T\hat{Q}_{13}\hat{P}_{33}\nonumber\\
&\hspace*{3.1cm}-B^T\hat{Q}_{13}\hat{P}_{22}+B^T\hat{Q}_{13}\hat{P}_{33}\big).\nonumber
\end{align}In general, $d_1\neq0$, $d_2\neq0$, and $d_3\neq0$. Therefore, it is inherently not possible to construct a local optimum of $||\Delta_{mul}(s)||_{\mathcal{H}_2}^2$ within the oblique projection framework. Nevertheless, the choice of $\bar{V}_r=P_{12}\hat{P}_{22}^{-1}$ and $\bar{W}_r=-\hat{Q}_{12}\hat{Q}_{22}^{-1}$ targets the optimality conditions (\ref{opc1})-(\ref{opc3}) and tends to achieve a local optimum of $||\Delta_{mul}(s)||_{\mathcal{H}_2}^2$. However, $\bar{V}_r$ and $\bar{W}_r$ with this choice depend on the ROM to be constructed, which is not known \textit{apriori}. The matrices $(\bar{A}_r,\bar{B}_r,\bar{C}_r)$ and $(\bar{V}_r,\bar{W}_r)$ can be seen as two coupled systems, i.e.,
\begin{align}
(\bar{A}_r,\bar{B}_r,\bar{C}_r)&=g(\bar{V}_r,\bar{W}_r)&\textnormal{and}&&(\bar{V}_r,\bar{W}_r)&=f(\bar{A}_r,\bar{B}_r,\bar{C}_r).\nonumber
\end{align}
The problem of constructing a ROM with the reduction matrices $\bar{V}_r=P_{12}\hat{P}_{22}^{-1}$ and $\bar{W}_r=-\hat{Q}_{12}\hat{Q}_{22}^{-1}$ can be seen as that of finding stationary points of
\begin{align}
(\bar{A}_r,\bar{B}_r,\bar{C}_r)=g\big(f(\bar{A}_r,\bar{B}_r,\bar{C}_r)\big)\nonumber
\end{align} with an additional constraint that $\bar{W}_r^T\bar{V}_r=I$. By starting with a random guess of $(\bar{A}_r,\bar{B}_r,\bar{C}_r)$, the stationary points can be obtained by using a stationary point iteration algorithm. The oblique projection condition $\bar{W}_r^T\bar{V}_r=I$ can be enforced in each iteration by using the bi-orthogonal Gram Smith method \cite{benner2011sparse}. However, the issue is that the solvability of $\hat{Q}_{12}$ and $\hat{Q}_{22}$ requires $\bar{H}_r(s)$ to be a minimum phase system in every iteration of the stationary point iteration algorithm, which cannot be generally guaranteed. We will shortly address this issue, but first, note that $\bar{V}_r=P_{12}$ and $\bar{W}_r=\hat{Q}_{12}$ span the same subspace as that spanned by $\bar{V}_r=P_{12}\hat{P}_{22}^{-1}$ and $\bar{W}_r=-\hat{Q}_{12}\hat{Q}_{22}^{-1}$ as $\hat{P}_{22}^{-1}$ and $-\hat{Q}_{22}^{-1}$ only change the basis, cf. Lemma 1.1 of \cite{gallivan2004sylvester}. The requirement for invertibility of $\hat{P}_{22}$ and $\hat{Q}_{22}$ in every iteration may cause numerical ill-conditioning or even failure. Therefore, $\bar{V}_r=P_{12}$ and $\bar{W}_r=\hat{Q}_{12}$ are a better choice for the reduction subspaces.

The optimality conditions (\ref{opc1})-(\ref{opc3}) are derived with an assumption that $\bar{H}_r(s)$ is a minimum phase system. If this condition is violated in an iteration, $\hat{Q}_{12}$ and $\hat{Q}_{22}$ cannot be obtained by solving the equations (\ref{v2.26}) and (\ref{v2.28}), respectively.

Let $\bar{G}_r(s)$ be a minimum phase right spectral factor of $\bar{H}_r(s)\bar{H}_r^*(s)$, such that $\bar{G}_r^*(s)\bar{G}_r(s)=\bar{H}_r(s)\bar{H}_r^*(s)$. Then $||\Delta_{mul}(s)||_{\mathcal{H}_2}$ can be redefined as
  \begin{align}
  ||\Delta_{mul}(s)||_{\mathcal{H}_2}
  &=\sqrt{\frac{1}{2\pi}\int_{-\infty}^{\infty}trace\Big(\Delta_{add}^*(j\omega)\bar{G}_r^{-1}(j\omega)\bar{G}_r^{-*}(j\omega)\Delta_{add}(j\omega)\Big)d\omega}.\nonumber
  \end{align}
It can readily be noted that $\bar{H}_r^{-1}(s)$ can be replaced with $\bar{G}_r^{-*}(s)$ to deal with a non minimum phase $\bar{H}_r(s)$. Next, we need to find a state-space realization of $\bar{G}_r^{-*}(s)$.

Let $\hat{P}_{11}$ be the controllability gramian of the pair $(\bar{A}_r,\bar{B}_r)$, which satisfies the following Lyapunov equation
   \begin{align}
   \bar{A}_r\hat{P}_{11}+\hat{P}_{11}\bar{A}_r^T+\bar{B}_r\bar{B}_r^T=0.\nonumber
   \end{align}
   Further, define $\bar{B}_x$ and $\bar{A}_x$ as the following
  \begin{align}
  \bar{B}_x&=\hat{P}_{11}\bar{C}_r^T+\bar{B}_rD^T,&\bar{A}_x&=\bar{A}_r-\bar{B}_x(DD^T)^{-1}\bar{C}_r.\nonumber
  \end{align}
  Let $\hat{X}$ solves the following Riccati equation
  \begin{align}
  \bar{A}_x^T\hat{X}+\hat{X}\bar{A}_x+\hat{X}\bar{B}_x(DD^T)^{-1}\bar{B}_x^T\hat{X}+\bar{C}_r^T(DD^T)^{-1}\bar{C}_r=0.\label{nnee38}
  \end{align}
  Then a possible realization of $\hat{G}_r(s)$ can be computed as
\begin{align}
\bar{A}_{sp}&=\bar{A}_r,&\bar{B}_{sp}&=\bar{B}_x,&\bar{C}_{sp}&=D^{-1}(\bar{C}_r-\bar{B}_x^T\hat{X}),&\bar{D}_{sp}&=D^T.\label{Req21}
\end{align}
Next, a possible minimum phase realization of $\bar{G}_r^*(s)$ is given as
  \begin{align}
  \bar{A}_g&=-\bar{A}_{sp}^T,&\bar{B}_g&=-\bar{C}_{sp}^T,&\bar{C}_g&=\bar{B}_{sp}^T,&\bar{D}_g=\bar{D}_{sp}^T.\label{nnee39}
  \end{align}
  Further, a possible stable realization of $\bar{G}_r^{-*}(s)$ is given as
  \begin{align}
\bar{A}_{g_i}&=\bar{A}_g-\bar{B}_g\bar{D}_g^{-1}\bar{C}_g,& \bar{B}_{g_i}&=-\bar{B}_g\bar{D}_g^{-1}, & \bar{C}_{g_i}&=\bar{D}_g^{-1}\bar{C}_g,& \bar{D}_{g_i}&=\bar{D}_g^{-1}.\label{nnee40}
\end{align}
We are now in a position to state the algorithm, which is named ``Iterative Relative-error $\mathcal{H}_2$-MOR Algorithm (IRHMORA)". The pseudo-code of IRHMORA is given in Algorithm \ref{Alg1}.
\begin{algorithm}[!h]
\textbf{Input:} Original system: $(A,B,C,D)$; Initial guess: $(\bar{A}_r,\bar{B}_r,\bar{C}_r)$\\
\textbf{Output:} ROM $(\bar{A}_r,\bar{B}_r,\bar{C}_r)$.
  \begin{algorithmic}[1]
      \STATE \hspace*{0.5cm}\textbf{if} ($rank[D]<m$)
      \STATE $D=\epsilon I$.\label{stp2}
      \STATE \hspace*{0.5cm}\textbf{end if}
      \STATE $i_k=0$.
      \STATE \hspace*{0.5cm}\textbf{while} (not converged) \textbf{do}
      \STATE $i_k=i_k+1$.
      \STATE Compute $\hat{X}$ from the equation (\ref{nnee38}).\label{stp7}
      \STATE Compute $(\bar{A}_{gi},\bar{B}_{gi},\bar{C}_{gi},\bar{D}_{gi})$ from (\ref{Req21})-(\ref{nnee40}).\label{stp8}
       \STATE Compute $\hat{P}_{12}$ from the equation (\ref{nnee7}).\label{stp9}
      \STATE Compute $\bar{Q}_{33}$, $\bar{Q}_{13}$, $\bar{Q}_{23}$, and $\bar{Q}_{12}$ by solving\label{stp10}\\
      $\bar{A}_{gi}^T\bar{Q}_{33}+\bar{Q}_{33}\bar{A}_{gi}+\bar{C}_{gi}^T\bar{C}_{gi}=0,$\\
      $A^T\bar{Q}_{13}+\bar{Q}_{13}\bar{A}_{gi}+C^T(\bar{B}_{gi}^T\bar{Q}_{33}+\bar{D}_{gi}^T\bar{C}_{gi})=0,$\\
      $\bar{A}_r^T\bar{Q}_{23}+\bar{Q}_{23}\bar{A}_{gi}-\bar{C}_r^T(\bar{B}_{gi}^T\bar{Q}_{33}+\bar{D}_{gi}^T\bar{C}_{gi})=0,$\\
      $A^T\bar{Q}_{12}+\bar{Q}_{12}\bar{A}_r+C^T\bar{B}_{gi}^T\bar{Q}_{23}^T-\bar{Q}_{13}\bar{B}_{gi}\bar{C}_r-C^T\bar{D}_{gi}^T\bar{D}_{gi}\bar{C}_r=0$.
      \STATE $\bar{V}_r=\hat{P}_{12}$, $\bar{W}_r=\bar{Q}_{12}$, and bi-orthogonalize to ensure $\bar{W}_r^T\bar{V}_r=I$; cf. \cite{benner2011sparse}.\label{Rstp11}
      \STATE $\bar{A}_r=\bar{W}_r^TA\bar{V}_r$, $\bar{B}_r=\bar{W}_r^TB$, $\bar{C}_r=C\bar{V}_r$.\label{stp17}
      \STATE \hspace*{1.5cm}\textbf{if} ($i_k=$ allowable number of iterations)
      \STATE \hspace*{0.5cm}Break loop.\label{stp19}
      \STATE \hspace*{1.5cm}\textbf{end if}
      \STATE \hspace*{0.5cm}\textbf{end while}
  \end{algorithmic}
  \caption{IRHMORA}\label{Alg1}
\end{algorithm}
Step (\ref{stp2}) replaces the original $D$ matrix with an artificial one if it is rank deficient. Steps (\ref{stp7}) and (\ref{stp8}) compute a stable realization of $\bar{G}_r^{-1}(s)$. Steps (\ref{stp9}) and (\ref{stp10}) compute the reduction matrices $\bar{V}_r$ and $\bar{W}_r$. Step (\ref{Rstp11}) enforces the condition $\bar{W}_r^T\bar{V}_r=I$. Step (\ref{stp17}) updates the ROM, and Step (\ref{stp19}) stops the algorithm prematurely if it does not converge.
\subsubsection{Choice of Free Parameters}
In \cite{zulfiqar2021frequency}, it is proposed to select the initial ROM such that it contains the dominant poles (having large residues) of $W(s)$ and $H(s)$, as it ensures that $||W(s)\Delta_{add}(s)||_{\mathcal{H}_2}$ is small to start with. In the problem under consideration, $W(s)=\bar{H}_r^{-1}(s)$ is not known \textit{apriori}. Nevertheless, heuristically, it is suggested to select an initial ROM, which contains dominant poles (with large residues) and/or zeros of $H(s)$. Such a ROM can be generated by using the efficient algorithms in \cite{rommes2006efficient} and \cite{martins2007computation}. The value of $\epsilon$ to generate an artificial $D$-matrix should be close to zero but enough to ensure the invertibility of $D$.
\subsubsection{Stopping Criterion}
The convergence of IRHMORA is not guaranteed, as is the case in most of the $\mathcal{H}_2$ MOR methods. To ensure the computational efficiency of IRHMORA in case it does not converge, it can be stopped after a maximum number of allowable iterations.
\subsubsection{Computational Savings}
The main advantage of IRHMORA is that it does not require solutions of large-scale Lyapunov and Riccati equations, unlike the situation in BST and IFWHMORA. Instead, it requires similar small-scale Lyapunov and Riccati equations, which can be computed cheaply. Further, the reduction matrices are computed by solving (sparse-dense) Sylvester equations, which typically arise in most of the $\mathcal{H}_2$-optimal MOR algorithms, cf. \cite{benner2011sparse}. Fortunately, an efficient algorithm for this type of Sylvester equation is proposed in \cite{benner2011sparse}, which can compute the reduction matrices of IRHMORA. Therefore, IRHMORA is computationally much more efficient than BST and IFWHMORA.
\section{Numerical Examples}
In this section, IRHMORA is tested on three numerical examples. The first two examples are SISO systems, while the third one is a MIMO system.
\subsection{Experimental Setup}
IRHMORA is compared with the Two-sided Iteration Algorithm (TSIA), which constructs a local optimum of $||\Delta_{add}(s)||_{\mathcal{H}_2}^2$ upon convergence, and BST, which is among the most accurate algorithm in the literature for the relative error criteria I and II. IRHMORA is also compared with BT, which is among the most accurate MOR algorithms. TSIA and IRHMORA are initialized arbitrarily with the same initial guess using MATLAB's ``rss" function. Since the $D$-matrix in all the examples is $0$, $\epsilon$ is set to $0.001$. With this value of $\epsilon$, the Riccati equations in IRHMORA and BST were solved without any problems in our experiments. The effect of $\epsilon$ on the accuracy of the ROM is not noticeable when it is close to zero. Thus we have not included results with different values of $\epsilon$. The maximum number of iterations for TSIA and IRHMORA is set to $200$. However, in all our experiments, both algorithms converged within $50$ iterations for all the examples.

For each example, an $\mathcal{H}_\infty$ controller $K(s)$ is designed using the loop shaping procedure proposed in \cite{mcfarlane1992loop}. The bandwidth $[0,\omega_0]$ rad/sec in each example is selected arbitrarily for demonstration purposes. In loop shaping controller design, a controller is designed, which ensures that the loop gain of $\bar{H}_r(s)K(s)$ is high before the crossover frequency $\omega_0$ for good noise attenuation in that region. Further, the loop gain of $\bar{H}_r(s)K(s)$ is desired to be low after the crossover frequency for good robust stability in the presence of uncertainties. Generally, the desired loop shape $\omega_0/s$ satisfies this condition. Since the loop gain $\bar{H}_r(s)K(s)$ is high within the bandwidth, $I+\bar{H}_r(s)K(s)\approx \bar{H}_r(s)K(s)$ and $\bar{H}_r(s)K(s)[I+\bar{H}_r(s)K(s)]^{-1}\approx I$ in that region. Further, since the loop gain $\bar{H}_r(s)K(s)$ is low outside the bandwidth, $I+H(s)K(s)\approx I$ in that region. Recall from the introduction section that to ensure that $K(s)$ is also a stabilizing controller for $H(s)$, $H(s)$ should be reduced using the following criterion
\begin{align}
\underset{\substack{\bar{H}_r(s)\\\textnormal{order}=r}}{\text{min}}||[\Delta_{add}(s)\bar{H}_r^{-1}(s)]\bar{H}_r(s)K(s)[I+\bar{H}_r(s)K(s)]^{-1}||,
\end{align} cf. \cite{obinata2012model}. Also, recall that this criterion has the following equivalent representation
\begin{align}
\underset{\substack{\bar{H}_r(s)\\\textnormal{order}=r}}{\text{min}}||[I+K(s)\bar{H}_r(s)]^{-1}K(s)\bar{H}_r(s)\Delta_{mul}(s)||,
\end{align} cf. \cite{ennth}. It can readily be noted that with loop shaping controller design, the plant reduction problem reduces to relative error MOR problem. In this section, it is highlighted with numerical results that IRHMORA provides good accuracy with respect to both $\Delta_{rel}(s)$ and $\Delta_{mul}(s)$, similar to BST. The closed-loop robust stability performance of BT, BST, TSIA, and IRHMORA is also compared after designing $\mathcal{H}_\infty$ controllers for the reduced-order plants constructed by these algorithms. To that end, the actual loop shape achieved with the original high-order plant, i.e., $H(s)K(s)$ and the robust stability measure $[I+K(s)\bar{H}_r(s)]^{-1}K(s)\bar{H}_r(s)\Delta_{mul}(s)||_{\mathcal{H}_\infty}$ are compared. The experiments are performed using MATLAB R2016a on a laptop with a $2$GHz Intel processor and $16$GB memory.
\subsection{Clamped Beam}
Clamped beam model is a $348^{th}$ order SISO system taken from the benchmark collection of dynamical systems for testing MOR algorithms, cf. \citep{chahlaoui2005benchmark}. The ROMs of orders $15-40$ are constructed using BT, BST, TSIA, and IRHMORA, and the respective $||\Delta_{mul}(s)||_{\mathcal{H}_2}$ are tabulated in Table \ref{tab1}.
\begin{table}[!h]
\centering
\caption{$\mathcal{H}_2$ norm of the Relative Error II, i.e., $||\Delta_{mul}(s)||_{\mathcal{H}_2}$}\label{tab1}
\begin{tabular}{|c|c|c|c|c|}
\hline
Order & BT & BST     & TSIA     & IRHMORA \\ \hline
15     & 2.3147 & 1.6905 & 2.8127 & 1.1913 \\ \hline
20     & 0.9632& 0.8097 & 0.9278  & 0.7954 \\ \hline
25     & 0.0589& 0.0559 & 0.0615  & 0.0517 \\ \hline
30     & 0.0235   & 0.0228 & 0.0253  & 0.0211  \\ \hline
35     & 0.0012   & 0.0011  & 0.0017  & 0.0009 \\ \hline
40    & 0.0011   & 0.0010  & 0.0012   & 0.0005  \\ \hline
\end{tabular}
\end{table}It can be noticed that IRHMORA ensures the least error in the $\mathcal{H}_2$ norm. In this example, BT has offered slightly better accuracy in terms of relative error than TSIA showing that the ROMs constructed by BT are more accurate than the local optima captured by TSIA. The singular values of $H(s)$ and the $20^{th}$ order ROM are plotted in Figure \ref{fig1}, which is a good visual tool for accessing $||\Delta_{mul}(s)||_{\mathcal{H}_\infty}$ in the frequency domain.
\begin{figure}[!h]
  \centering
  \includegraphics[width=8cm]{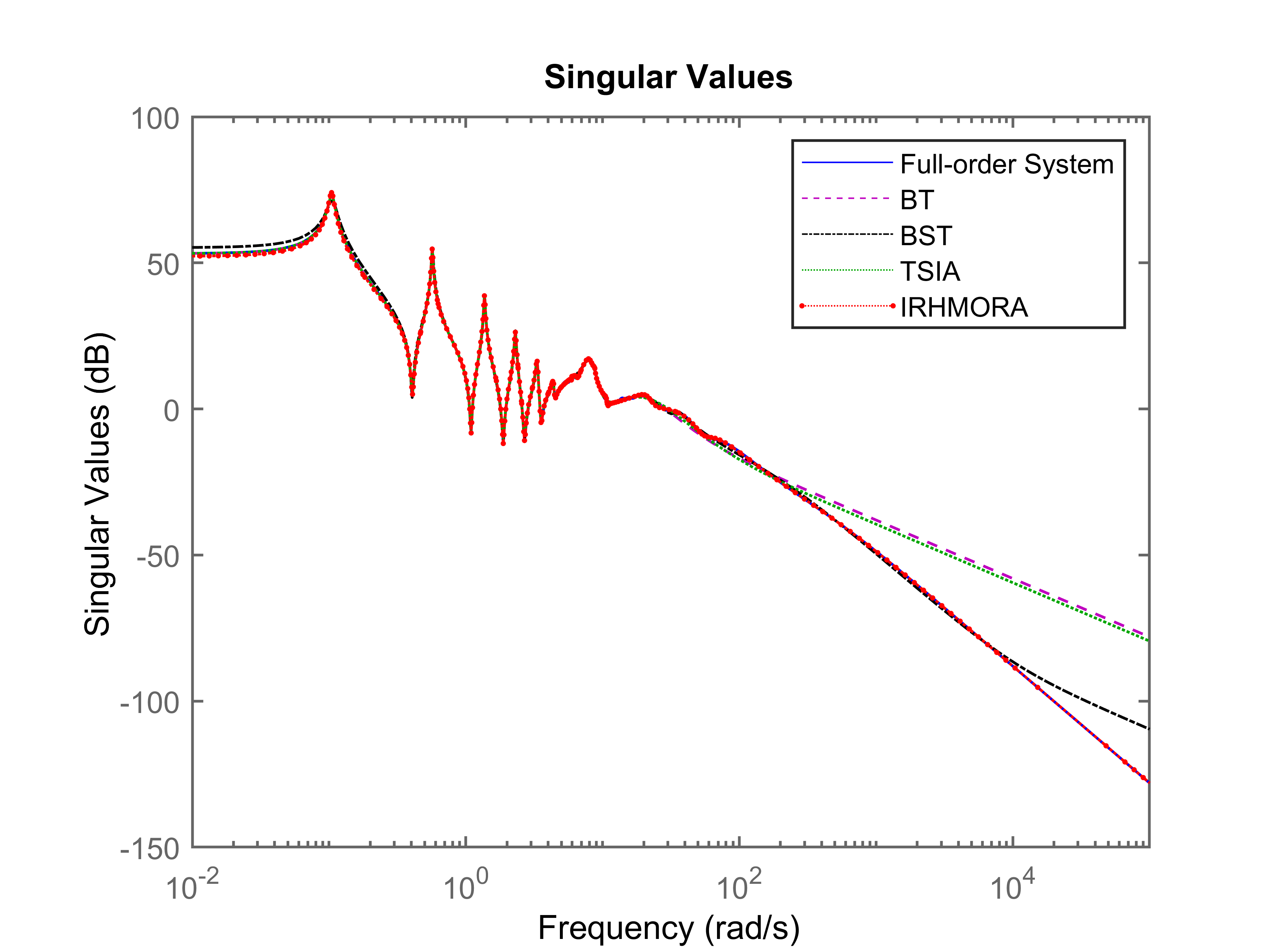}
  \caption{Sigma plots of the original and reduced models}\label{fig1}
\end{figure}It can be seen that IRHMORA offers good accuracy. Since $\Delta_{rel}(s)$ and $\Delta_{mul}(s)$ are closely related, the singular values of $\Delta_{rel}(s)$ for $15^{th}$ order ROMs are plotted in Figure \ref{fig2} for comparison.
\begin{figure}[!h]
  \centering
  \includegraphics[width=8cm]{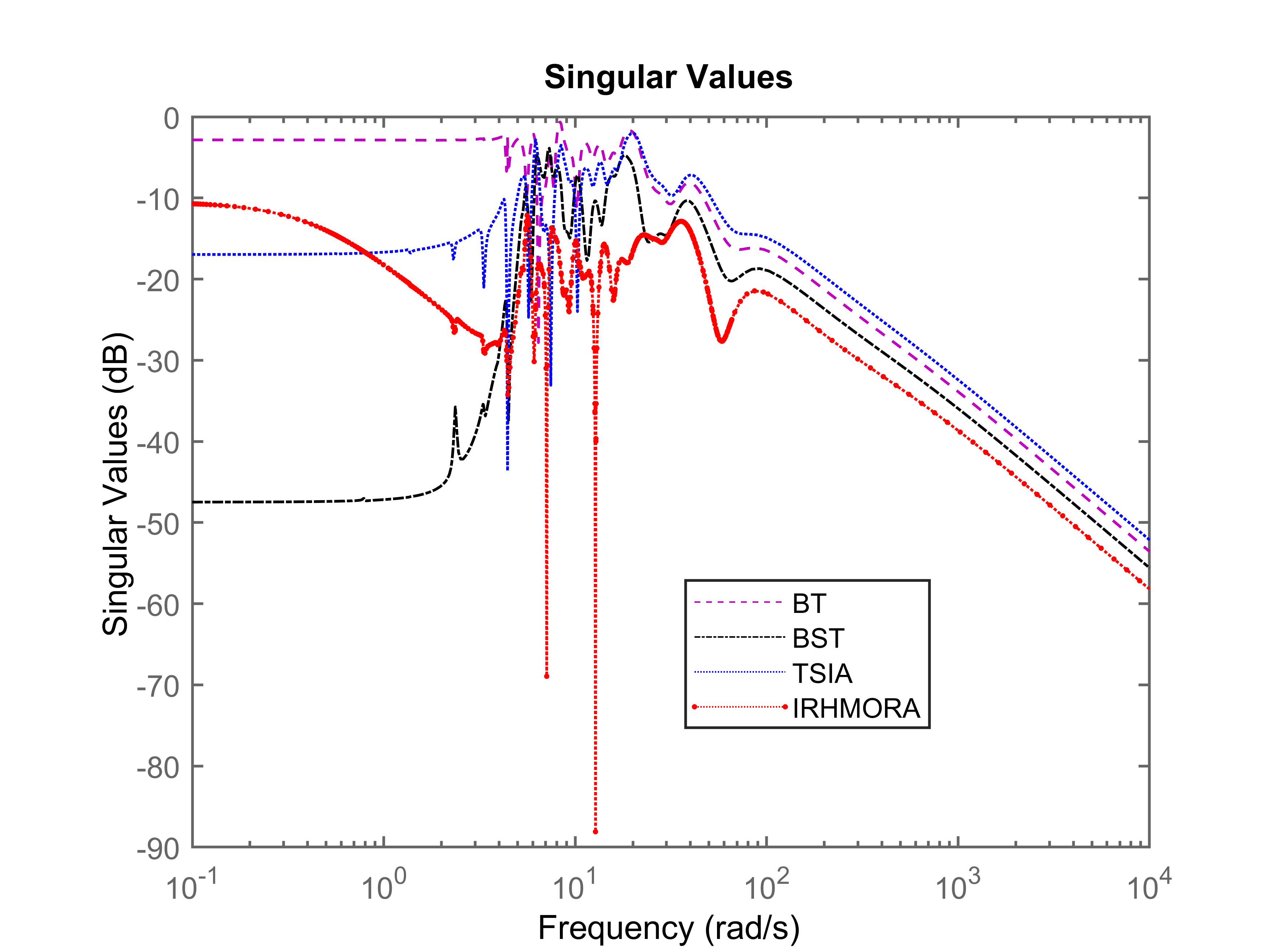}
  \caption{Sigma plot of the relative error I $\Delta_{rel}(s)$}\label{fig2}
\end{figure} Again, it can be seen that IRHMORA and BST offer good accuracy. In terms of the worst-case scenario, it can be noted from Figure \ref{fig2} that BST has done a slightly better job than IRHMORA since BST tends to minimize the $\mathcal{H}_\infty$ norm of $\Delta_{rel}(s)$.

For demonstration purposes, the crossover frequency for designing the $\mathcal{H}_\infty$ controller via loop shaping is selected as $5$ rad/sec. Accordingly, the desired loop shape is selected as $5/s$. An $18^{th}$-order $\mathcal{H}_\infty$ controller is designed using the loop shaping procedure proposed in \cite{mcfarlane1992loop}. The $15^{th}$-order ROMs are used as plant models for the controllers designed. The actual loop shapes achieved with these controllers and the full-order plant are plotted in Figure \ref{figR3}.
\begin{figure}[!h]
  \centering
  \includegraphics[width=8cm]{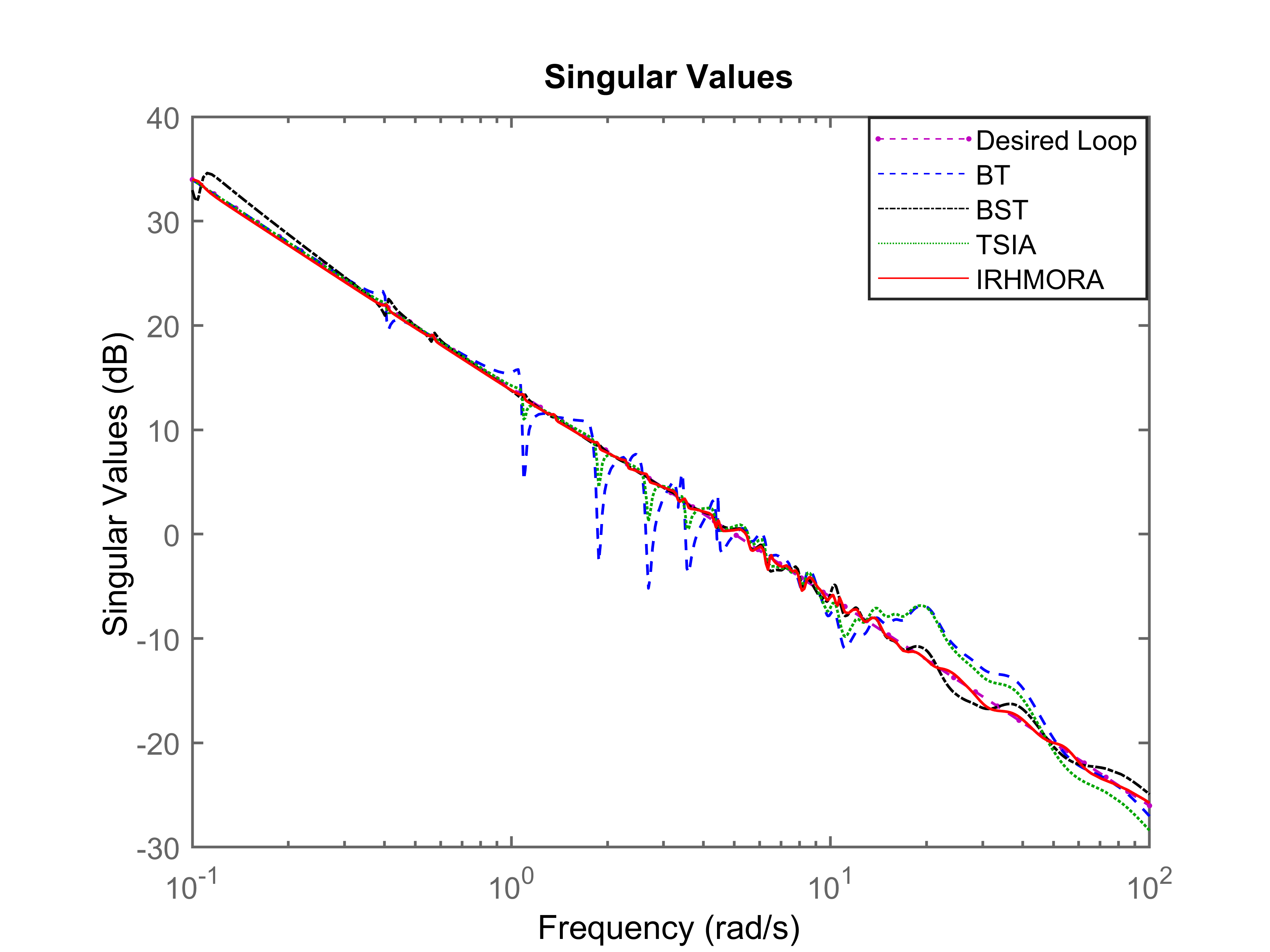}
  \caption{Loop Shape $H(s)K(s)$}\label{figR3}
\end{figure}
It can be seen that the loop shapes achieved by using the ROMs constructed by BST and IRHMORA are closest to the desired loop shape. The robust stability measure of the ROMs is tabulated in Table \ref{tabR1}.
\begin{table}[!h]
\centering
\caption{Robust Stability Measure}\label{tabR1}
\begin{tabular}{|c|c|c|c|c|}
\hline
Method & $||[I+K(s)\bar{H}_r(s)]^{-1}K(s)\bar{H}_r(s)\Delta_{mul}(s)||_{\mathcal{H}_\infty}$ \\ \hline
BT     & 0.6557  \\ \hline
BST     & 0.0357\\ \hline
TSIA     & 0.8491\\ \hline
IRHMORA    & 0.0240\\ \hline
\end{tabular}
\end{table}It can be seen that the ROMs constructed by BST and IRHMORA and the subsequent controllers designed for these models provide the best robust stability when connected in the closed loop. The reason for this superior performance is that the additive error is an open-loop reduction criterion, whereas the relative error is inherently a closed-loop reduction criterion.
\subsection{Artificial Dynamical System}
Artificial dynamical system model is a $1006^{th}$ order SISO system taken from the benchmark collection of dynamical systems for testing MOR algorithms, cf. \citep{chahlaoui2005benchmark}. The ROMs of orders $15-40$ are constructed using BT, BST, TSIA, and IRHMORA, and the respective $||\Delta_{mul}(s)||_{\mathcal{H}_2}$ are tabulated in Table \ref{tab2}.
\begin{table}[!h]
\centering
\caption{$\mathcal{H}_2$ norm of the Relative Error II, i.e., $||\Delta_{mul}(s)||_{\mathcal{H}_2}$}\label{tab2}
\begin{tabular}{|c|c|c|c|c|}
\hline
Order &BT& BST     & TSIA     & IRHMORA \\ \hline
15    &0.7246 & 0.0055 & 0.0457 & 0.0043 \\ \hline
20    &0.0081 & $9.058\times 10^{-4}$ & 0.0062  &  $3.913\times 10^{-4}$\\ \hline
25    &$6.324\times 10^{-4}$& $5.750\times 10^{-5}$ & $2.785\times 10^{-4}$  & $4.569\times 10^{-5}$ \\ \hline
30    &$1.576\times 10^{-4}$ &$4.854\times 10^{-5}$ & $1.297\times 10^{-4}$ & $1.387\times 10^{-5}$  \\ \hline
35    &$1.411\times 10^{-4}$ &$4.757\times 10^{-5}$ & $1.218\times 10^{-4}$ & $1.292\times 10^{-5}$  \\ \hline
40    &$1.134\times 10^{-4}$ &$2.941\times 10^{-5}$ & $1.029\times 10^{-4}$ & $1.005\times 10^{-5}$  \\ \hline
\end{tabular}
\end{table}
It can be noticed that IRHMORA ensures the least error in the $\mathcal{H}_2$ norm. The singular values of $H(s)$ and the $15^{th}$ order ROM are plotted in Figure \ref{fig3}, which is a good visual tool for accessing $||\Delta_{mul}(s)||_{\mathcal{H}_\infty}$ in the frequency domain.
\begin{figure}[!h]
  \centering
  \includegraphics[width=8cm]{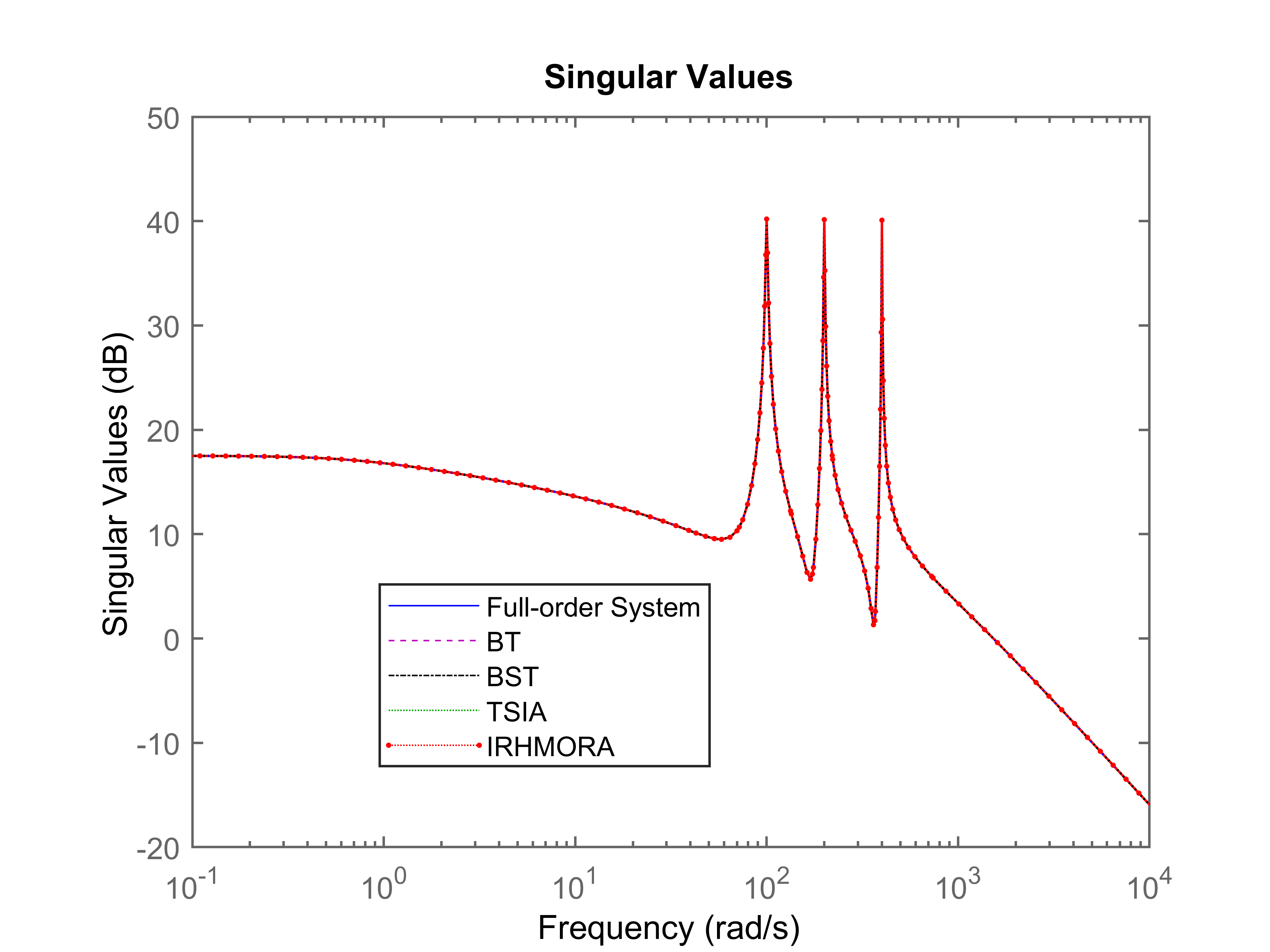}
  \caption{Sigma plots of the original and reduced models}\label{fig3}
\end{figure}
It can be seen that IRHMORA offers good accuracy. The singular values of $\Delta_{rel}(s)$ for $20^{th}$ order ROM are plotted in Figure \ref{fig4} for comparison.
\begin{figure}[!h]
  \centering
  \includegraphics[width=8cm]{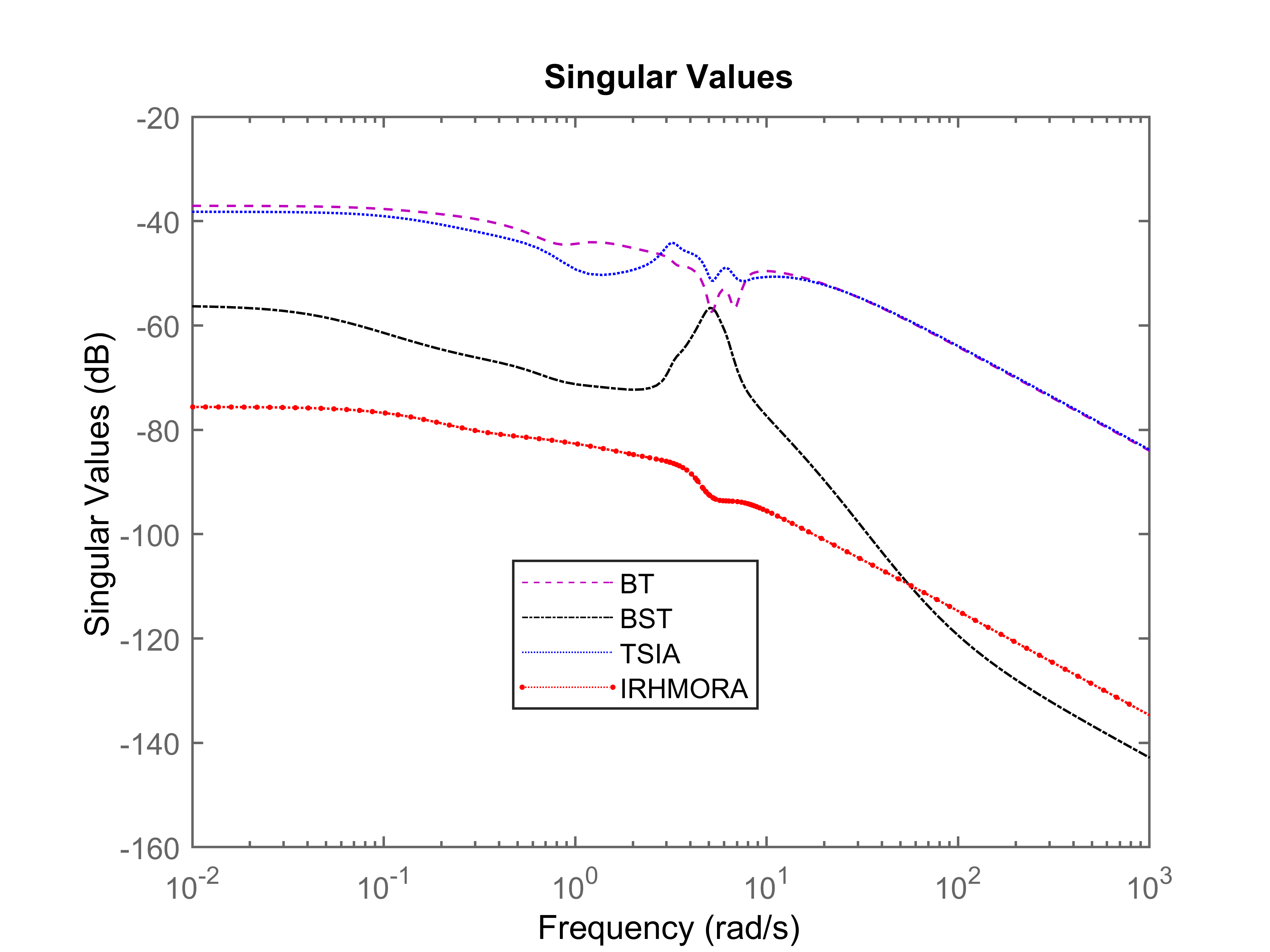}
  \caption{Sigma plot of the relative error I $\Delta_{rel}(s)$}\label{fig4}
\end{figure}
Again, it can be seen that IRHMORA offers good accuracy.

For demonstration purposes, the crossover frequency for designing the $\mathcal{H}_\infty$ controller via loop shaping is selected as $10$ rad/sec. Accordingly, the desired loop shape is selected as $10/s$. A $23^{th}$-order $\mathcal{H}_\infty$ controller is designed using the loop shaping procedure proposed in \cite{mcfarlane1992loop}. The $20^{th}$-order ROMs are used as plant models for the controllers designed. The actual loop shapes achieved with these controllers and the full-order plant are plotted in Figure \ref{figR4}.
\begin{figure}[!h]
  \centering
  \includegraphics[width=8cm]{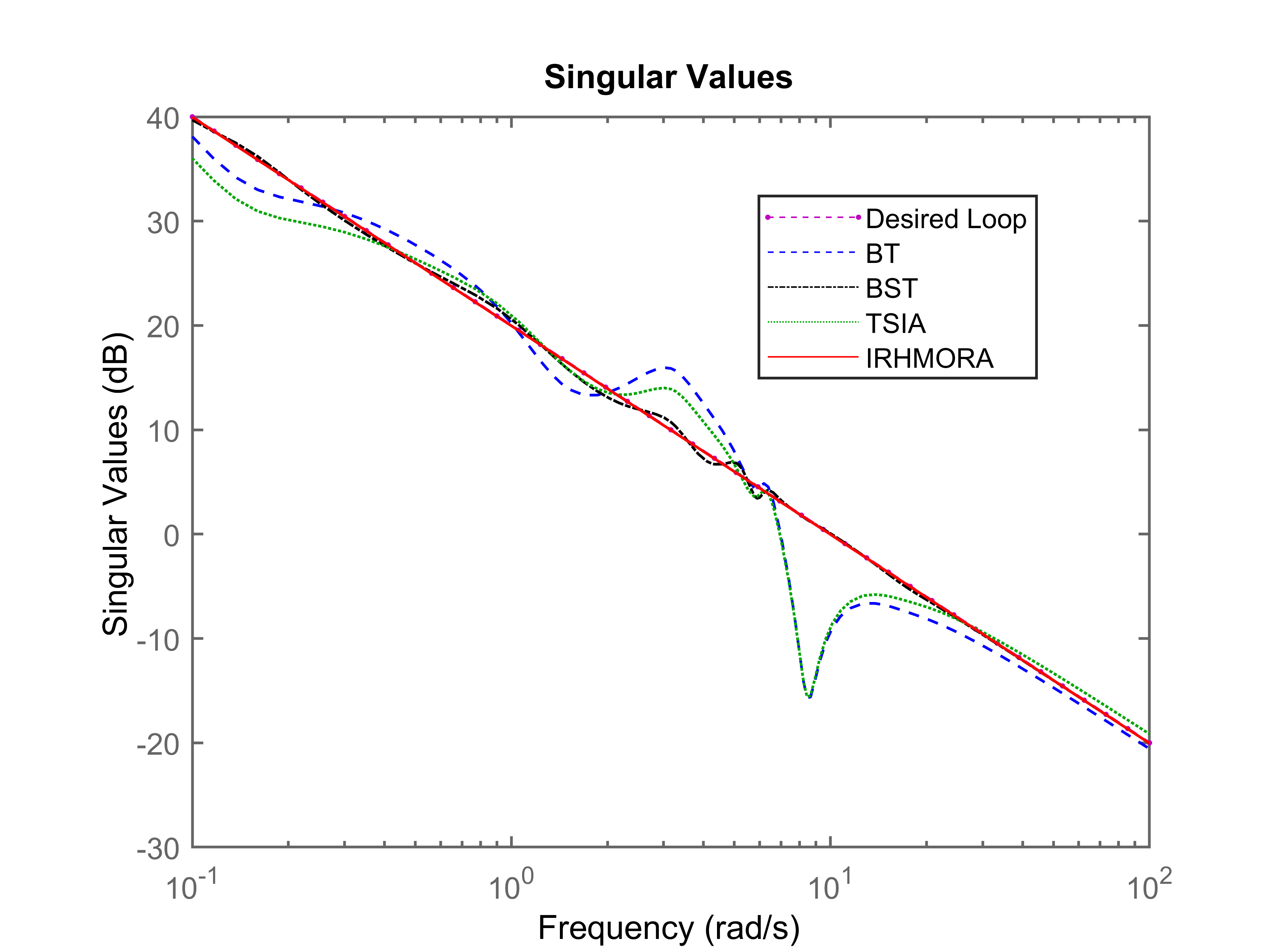}
  \caption{Loop Shape $H(s)K(s)$}\label{figR4}
\end{figure}
It can be seen that the loop shapes achieved by using the ROMs constructed by BST and IRHMORA are closest to the desired loop shape. The robust stability measure of the ROMs is tabulated in Table \ref{tabR2}.
\begin{table}[!h]
\centering
\caption{Robust Stability Measure}\label{tabR2}
\begin{tabular}{|c|c|c|c|c|}
\hline
Method & $||[I+K(s)\bar{H}_r(s)]^{-1}K(s)\bar{H}_r(s)\Delta_{mul}(s)||_{\mathcal{H}_\infty}$ \\ \hline
BT     & 1.5141  \\ \hline
BST     & 0.4253\\ \hline
TSIA     & 1.7906\\ \hline
IRHMORA    & 0.2954\\ \hline
\end{tabular}
\end{table}It can be seen that the ROMs constructed by BST and IRHMORA and the subsequent controllers designed for these models provide the best robust stability when connected in the closed loop.
\subsection{Brazil Interconnected Power System}
Brazil interconnected power system model is a $3077^{th}$ order MIMO system (with four inputs and four outputs) taken from \cite{rommes2009computing}. The ROMs of orders $15-40$ are constructed using BT, BST, TSIA, and IRHMORA, and the respective $||\Delta_{mul}(s)||_{\mathcal{H}_2}$ are tabulated in Table \ref{tab3}.
\begin{table}[!h]
\centering
\caption{$\mathcal{H}_2$ norm of the Relative Error II, i.e., $||\Delta_{mul}(s)||_{\mathcal{H}_2}$}\label{tab3}
\begin{tabular}{|c|c|c|c|c|}
\hline
Order &BT& BST     & TSIA     & IRHMORA \\ \hline
15     & 7.6929&6.5678 & 7.1019 & 5.3172 \\ \hline
20     & 5.9076&5.1848 & 5.4515 & 2.5768 \\ \hline
25     & 1.7864&1.2152 & 1.5220 & 0.8525 \\ \hline
30     & 0.9253&0.7418 & 0.9126 & 0.5693 \\ \hline
35     & 0.6012&0.5276 & 0.5551 & 0.2092 \\ \hline
40     & 0.5868&0.3102 & 0.4998 & 0.1509 \\ \hline
\end{tabular}
\end{table}
It can be noticed that IRHMORA ensures the least error. The computational time to construct the $15^{th}$-order ROM is tabulated in Table \ref{tabS}.
\begin{table}[!h]
\centering
\caption{Computational Time Comparison}\label{tabS}
\begin{tabular}{|c|c|c|c|c|}
\hline
Method & Time (sec)\\ \hline
BT     & 173.29  \\ \hline
BST     & 4330.96\\ \hline
TSIA     & 27.15\\ \hline
IRHMORA    & 36.66\\ \hline
\end{tabular}
\end{table} It can be seen that TSIA and IRHMORA take much less computational time to construct the ROM because they do not require the solution of high-order dense Lyapunov equations.

For demonstration purposes, the crossover frequency for designing the $\mathcal{H}_\infty$ controller via loop shaping is selected as $3$ rad/sec. Accordingly, the desired loop shape is selected as $3/s$. A $45^{th}$-order $\mathcal{H}_\infty$ controller is designed using the loop shaping procedure proposed in \cite{mcfarlane1992loop}. The $40^{th}$-order ROMs are used as plant models for the controllers designed. The actual loop shapes achieved with these controllers and the full-order plant are plotted in Figure \ref{figR5}.
\begin{figure}[!h]
  \centering
  \includegraphics[width=8cm]{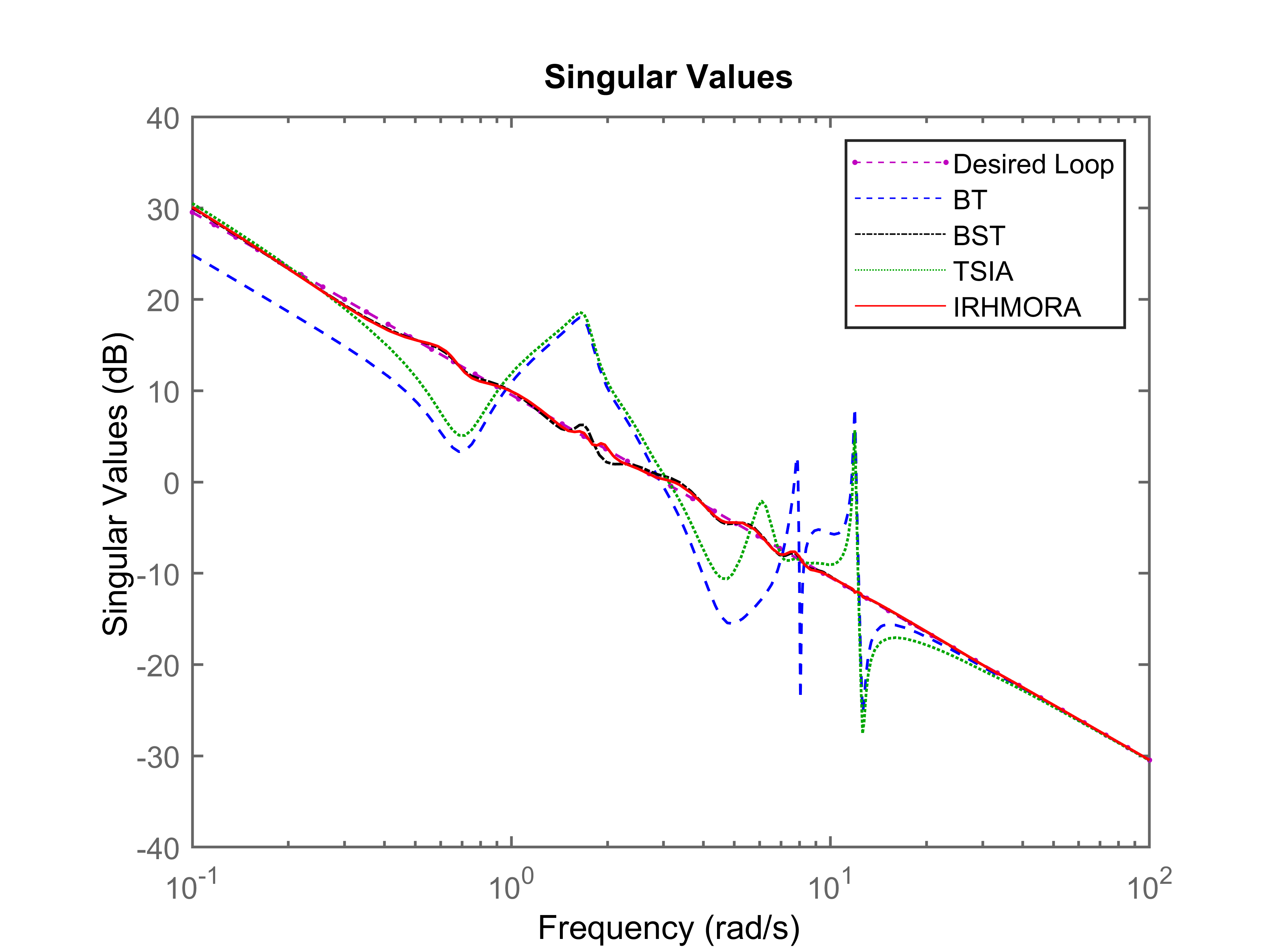}
  \caption{Loop Shape $H(s)K(s)$}\label{figR5}
\end{figure}
It can be seen that the loop shapes achieved by using the ROMs constructed by BST and IRHMORA are closest to the desired loop shape. The robust stability measure of the ROMs is tabulated in Table \ref{tabR3}.
\begin{table}[!h]
\centering
\caption{Robust Stability Measure}\label{tabR3}
\begin{tabular}{|c|c|c|c|c|}
\hline
Method & $||[I+K(s)\bar{H}_r(s)]^{-1}K(s)\bar{H}_r(s)\Delta_{mul}(s)||_{\mathcal{H}_\infty}$ \\ \hline
BT     & 5.6895  \\ \hline
BST     & 0.9440\\ \hline
TSIA     & 4.7935\\ \hline
IRHMORA    & 0.8982\\ \hline
\end{tabular}
\end{table}It can be seen that the ROMs constructed by BST and IRHMORA and the subsequent controllers designed for these models provide the best robust stability when connected in the closed loop.
\section{Conclusion}
The problem of relative error $\mathcal{H}_2$ MOR is addressed. It is shown that IFWHMORA can be used for this problem, but it requires solutions of large-scale Lyapunov/Sylvester and Riccati equations, which makes it computationally infeasible in a large-scale setting. A different relative error criterion is then considered, and the necessary conditions for a local optimum are derived. Based on these conditions, an oblique projection algorithm is proposed, which is computationally efficient. The proposed algorithm compares well in accuracy with BST, which is a gold standard for this problem in the literature. The significance of the proposed algorithm in designing reduced-order controllers for high-order plants is highlighted by considering benchmark numerical examples. The numerical results confirm the theory developed in the paper.
\section*{Appendix}
Let us define the cost function $\hat{J}$ as $\hat{J}=||\Delta_{mul}||_{\mathcal{H}_2}^2$ and denote the first-order derivative of $\hat{Q}_{11}$, $\hat{Q}_{12}$, $\hat{Q}_{13}$, $\hat{Q}_{22}$, $\hat{Q}_{23}$, $\hat{Q}_{33}$, and $\hat{J}$ with respect to $\bar{A}_r$ as $\Delta_{\hat{Q}_{11}}^{\bar{A}_r}$, $\Delta_{\hat{Q}_{12}}^{\bar{A}_r}$, $\Delta_{\hat{Q}_{13}}^{\bar{A}_r}$, $\Delta_{\hat{Q}_{22}}^{\bar{A}_r}$, $\Delta_{\hat{Q}_{23}}^{\bar{A}_r}$, $\Delta_{\hat{Q}_{33}}^{\bar{A}_r}$, and $\Delta_{\hat{J}}^{\bar{A}_r}$, respectively. Further, let us denote the differential of $\bar{A}_r$ as $\Delta_{\bar{A}_r}$. By differentiating the equations (\ref{v2.25})-(\ref{v2.30}) with respect to $\bar{A}_r$, we get
      \begin{align}
      A^T\Delta_{\hat{Q}_{11}}^{\bar{A}_r}+\Delta_{\hat{Q}_{11}}^{\bar{A}_r}A+\hat{S}_1&=0,\hspace*{2cm}\label{v2.40}\\
      A^T\Delta_{\hat{Q}_{12}}^{\bar{A}_r}+\Delta_{\hat{Q}_{12}}^{\bar{A}_r}\bar{A}_r+\hat{S}_2&=0\label{v2.41},\\
      A^T\Delta_{\hat{Q}_{13}}^{\bar{A}_r}+\Delta_{\hat{Q}_{13}}^{\bar{A}_r}(\bar{A}_r-\bar{B}_rD^{-1}\bar{C}_r)+\hat{S}_3&=0,\label{v2.42}\\
      \bar{A}_r^T\Delta_{\hat{Q}_{22}}^{\bar{A}_r}+\Delta_{\hat{Q}_{22}}^{\bar{A}_r}\bar{A}_r+\hat{S}_4&=0,\label{v2.43}\\
      \bar{A}_r^T\Delta_{\hat{Q}_{23}}^{\bar{A}_r}+\Delta_{\hat{Q}_{23}}^{\bar{A}_r}(\bar{A}_r-\bar{B}_rD^{-1}\bar{C}_r)+\hat{S}_5&=0,\label{v2.44}\\
      (\bar{A}_r-\bar{B}_rD^{-1}\bar{C}_r)^T\Delta_{\hat{Q}_{33}}^{\bar{A}_r}+\Delta_{\hat{Q}_{33}}^{\bar{A}_r}(\bar{A}_r-\bar{B}_rD^{-1}\bar{C}_r)+\hat{S}_6&=0\label{v2.45}
      \end{align}
      wherein
      \begin{align}
      \hspace*{1.5cm}\hat{S}_1&=-C^TD^{-T}\bar{B}_r^T(\Delta_{\hat{Q}_{13}}^{\bar{A}_r})^T-\Delta_{\hat{Q}_{13}}^{\bar{A}_r}\bar{B}_rD^{-1}C,\nonumber\\
      \hat{S}_2&=-C^TD^{-T}\bar{B}_r^T\Delta_{\hat{Q}_{23}}^{\bar{A}_r}+\Delta_{\hat{Q}_{13}}^{\bar{A}_r}\bar{B}_rD^{-1}\bar{C}_r+\hat{Q}_{12}\Delta_{\bar{A}_r},\nonumber\\
      \hat{S}_3&=\hat{Q}_{13}\Delta_{\bar{A}_r}-C^TD^{-T}\bar{B}_r^T\Delta_{\hat{Q}_{33}}^{\bar{A}_r},\nonumber\\
      \hat{S}_4&=(\Delta_{\bar{A}_r})^T\hat{Q}_{22}+\hat{Q}_{22}\Delta_{\bar{A}_r}+\bar{C}_r^TD^{-T}\bar{B}_r^T\Delta_{\hat{Q}_{23}}^{\bar{A}_r}+\Delta_{\hat{Q}_{23}}^{\bar{A}_r}\bar{B}_rD^{-1}\bar{C}_r,\nonumber\\
      \hat{S}_5&=(\Delta_{\bar{A}_r})^T\hat{Q}_{23}+\hat{Q}_{23}\Delta_{\bar{A}_r}+\bar{C}_r^TD^{-T}\bar{B}_r^T\Delta_{\hat{Q}_{33}}^{\bar{A}_r},\nonumber\\
      \hat{S}_6&=(\Delta_{\bar{A}_r})^T\hat{Q}_{33}+\hat{Q}_{33}\Delta_{\bar{A}_r}.\nonumber
      \end{align}
      By differentiating $\hat{J}$ with respect to $\bar{A}_r$, we get
      \begin{align}
      \Delta_{\hat{J}}^{\bar{A}_r}&=trace(B^T\Delta_{\hat{Q}_{11}}^{\bar{A}_r}B+2B^T\Delta_{\hat{Q}_{12}}^{\bar{A}_r}\bar{B}_r+\bar{B}_r^T\Delta_{\hat{Q}_{22}}^{\bar{A_r}}\bar{B}_r)\nonumber\\
      &=trace(BB^T\Delta_{\hat{Q}_{11}}^{\bar{A}_r}+2B\bar{B}_r^T(\Delta_{\hat{Q}_{12}}^{\bar{A}_r})^T+\bar{B}_r\bar{B}_r^T\Delta_{\hat{Q}_{22}}^{\bar{A}_r})\nonumber
      \end{align}
      By applying Lemma \ref{lemma} on the equations (\ref{v2.40}) and (\ref{nnee6}), (\ref{v2.41}) and (\ref{nnee7}), and (\ref{v2.43}) and (\ref{nnee8}), we get
      \begin{align}
      trace(BB^T\Delta_{\hat{Q}_{11}}^{\bar{A}_r})&=trace(\hat{S}_1P_{11}),\nonumber\\
      trace\big(B\bar{B}_r^T(\Delta_{\hat{Q}_{12}}^{\bar{A}_r})^T\big)&=trace(\hat{S}_2^TP_{12}),\nonumber\\
      trace(\bar{B}_r\bar{B}_r^T\Delta_{\hat{Q}_{22}}^{\bar{A}_r})&=trace(\hat{S}_4\hat{P}_{22}).\nonumber
      \end{align}
      Thus
      \begin{align}
      \Delta_{\hat{J}}^{\bar{A}_r}&=trace\Big(-C^TD^{-T}\bar{B}_r^T(\Delta_{\hat{Q}_{13}}^{\bar{A}_r})^TP_{11}-(\Delta_{\hat{Q}_{13}}^{\bar{A}_r})^T\bar{B}_rD^{-1}CP_{11}\nonumber\\
      &\hspace*{2.5cm}-2\Delta_{\hat{Q}_{23}}^{\bar{A}_r}\bar{B}_rD^{-1}CP_{12}+2\bar{C}_r^TD^{-T}\bar{B}_r^T(\Delta_{\hat{Q}_{13}}^{\bar{A}_r})^TP_{12}\nonumber\\
      &\hspace*{2.5cm}+2(\Delta_{\bar{A}_r})^T\hat{Q}_{12}^TP_{12}+(\Delta_{\bar{A}_r})^T\hat{Q}_{22}\hat{P}_{22}+\hat{Q}_{22}\Delta_{\bar{A}_r}\hat{P}_{22}\nonumber\\
      &\hspace*{2.5cm}+\bar{C}_r^TD^{-T}\bar{B}_r^T\Delta_{\hat{Q}_{23}}^{\bar{A}_r}\hat{P}_{22}+\Delta_{\hat{Q}_{23}}^{\bar{A}_r}\bar{B}_rD^{-T}\bar{C}_r\hat{P}_{22}\Big)\nonumber\\
      &=2trace\Big(\hat{Q}_{12}^TP_{12}(\Delta_{\bar{A}_r})^T+\hat{Q}_{22}\hat{P}_{22}(\Delta_{\bar{A}_r})^T+P_{12}\bar{C}_r^TD^{-T}\bar{B}_r^T(\Delta_{\hat{Q}_{13}}^{\bar{A}_r})^T\nonumber\\
      &\hspace*{2.5cm}-P_{11}C^TD^{-T}\bar{B}_r^T(\Delta_{\hat{Q}_{13}}^{\bar{A}_r})^T-P_{12}^TC^TD^{-T}\bar{B}_r^T\Delta_{\hat{Q}_{23}}^{\bar{A}_r}\nonumber\\
      &\hspace*{2.5cm}+\hat{P}_{22}\bar{C}_r^TD^{-T}\bar{B}_r^T\Delta_{\hat{Q}_{23}}^{\bar{A}_r}\Big).\nonumber
      \end{align}
      By applying Lemma \ref{lemma} on the equations (\ref{v2.42}) and (\ref{v2.21}), and (\ref{v2.44}) and (\ref{v2.22}), we get
      \begin{align}
      trace(\hat{S}_3^T\hat{P}_{13})&=trace\big(-P_{11}C^TD^{-T}\bar{B}_r^T(\Delta_{\hat{Q}_{13}}^{\bar{A}_r})^T+P_{12}\bar{C}_r^TD^{-T}\bar{B}_r^T(\Delta_{\hat{Q}_{13}}^{\bar{A}_r})^T\big)\nonumber\\
      trace(\hat{S}_5^T\hat{P}_{23})&=trace\big(-P_{12}^TC^TD^{-T}\bar{B}_r^T\Delta_{\hat{Q}_{23}}^{\bar{A}_r}+\hat{P}_{22}\bar{C}_r^TD^{-T}\bar{B}_r^T\Delta_{\hat{Q}_{23}}^{\bar{A}_r}\big).\nonumber
      \end{align}
      Thus
      \begin{align}
      \Delta_{\hat{J}}^{\bar{A}_r}&=2trace\Big(\hat{Q}_{12}^TP_{12}(\Delta_{\bar{A}_r})^T+\hat{Q}_{22}\hat{P}_{22}(\Delta_{\bar{A}_r})^T+\hat{Q}_{13}^T\hat{P}_{13}(\Delta_{\bar{A}_r})^T\nonumber\\
      &\hspace*{2.5cm}+\hat{Q}_{23}\hat{P}_{23}^T(\Delta_{\bar{A}_r})^T+\hat{Q}_{23}\hat{P}_{23}(\Delta_{\bar{A}_r})^T\nonumber\\
      &\hspace*{2.5cm}-\bar{B}_rD^{-1}C\hat{P}_{13}\Delta_{\hat{Q}_{33}}^{\bar{A}_r}+\bar{B}_rD^{-1}\bar{C}_r\hat{P}_{23}\Delta_{\hat{Q}_{33}}^{\bar{A}_r}\Big).\nonumber
      \end{align}
      By applying Lemma \ref{lemma} on the equations (\ref{v2.45}) and (\ref{v2.23}), we get
      \begin{align}
      trace(\hat{S}_6\hat{P}_{33})=2trace(-\bar{B}_rD^{-1}C\hat{P}_{13}\Delta_{\hat{Q}_{33}}^{\hat{A}_r}+\bar{B}_rD^{-1}\bar{C}_r\hat{P}_{23}\Delta_{\hat{Q}_{33}}^{\bar{A}_r}).\nonumber
      \end{align}
      Thus
       \begin{align}
      \Delta_{\hat{J}}^{\bar{A}_r}&=2trace\Big(\hat{Q}_{12}^TP_{12}(\Delta_{\bar{A}_r})^T+\hat{Q}_{22}\hat{P}_{22}(\Delta_{\bar{A}_r})^T+\hat{Q}_{13}^T\hat{P}_{13}(\Delta_{\bar{A}_r})^T\nonumber\\
      &\hspace*{2.5cm}+\hat{Q}_{23}\hat{P}_{23}^T(\Delta_{\bar{A}_r})^T+\hat{Q}_{23}\hat{P}_{23}(\Delta_{\bar{A}_r})^T+\hat{Q}_{33}\hat{P}_{33}(\Delta_{\bar{A}_r})^T\Big)\nonumber\\
      &=2trace\Big(\big(\hat{Q}_{12}^TP_{12}+\hat{Q}_{22}P_{22}+\hat{Q}_{13}^T\hat{P}_{13}+\hat{Q}_{23}\hat{P}_{23}^T+\hat{Q}_{23}^T\hat{P}_{23}+\hat{Q}_{33}\hat{P}_{33}\big)(\Delta_{\bar{A}_r})^T\Big)\nonumber
      \end{align} and
      \begin{align}
      \frac{\partial}{\partial\bar{A}_r}||\Delta_{mul}(s)||_{\mathcal{H}_2}^2=2(\hat{Q}_{12}^TP_{12}+\hat{Q}_{22}\hat{P}_{22}+\hat{Q}_{13}^T\hat{P}_{13}+\hat{Q}_{23}\hat{P}_{23}^T+\hat{Q}_{23}\hat{P}_{23}+\hat{Q}_{33}\hat{P}_{33}).\nonumber
      \end{align}
      Hence,
      \begin{align}
      \hat{Q}_{12}^TP_{12}+\hat{Q}_{22}\hat{P}_{22}+\bar{X}=0\nonumber
      \end{align} is a necessary condition for the local optimum of $||\Delta_{mul}(s)||_{\mathcal{H}_2}^2$.

  Let us denote the first-order derivative of $\hat{Q}_{11}$, $\hat{Q}_{12}$, $\hat{Q}_{13}$, $\hat{Q}_{22}$, $\hat{Q}_{23}$, $\hat{Q}_{33}$, and $\hat{J}$ with respect to $\bar{B}_r$ as $\Delta_{\hat{Q}_{11}}^{\bar{B}_r}$, $\Delta_{\hat{Q}_{12}}^{\bar{B}_r}$, $\Delta_{\hat{Q}_{13}}^{\bar{B}_r}$, $\Delta_{\hat{Q}_{22}}^{\bar{B}_r}$, $\Delta_{\hat{Q}_{23}}^{\bar{B}_r}$, $\Delta_{\hat{Q}_{33}}^{\bar{B}_r}$, and $\Delta_{\hat{J}}^{\bar{B}_r}$, respectively. Further, let us denote the differential of $\bar{B}_r$ as $\Delta_{\bar{B}_r}$. By differentiating the equations (\ref{v2.25})-(\ref{v2.30}) with respect to $\bar{B}_r$, we get
      \begin{align}
      A^T\Delta_{\hat{Q}_{11}}^{\bar{B}_r}+\Delta_{\hat{Q}_{11}}^{\bar{B}_r}A+\hat{R}_1&=0,\hspace*{2cm}\label{v2.46}\\
      A^T\Delta_{\hat{Q}_{12}}^{\bar{B}_r}+\Delta_{\hat{Q}_{12}}^{\bar{B}_r}\bar{A}_r+\hat{R}_2&=0,\label{v2.47}\\
      A^T\Delta_{\hat{Q}_{13}}^{\bar{B}_r}+\Delta_{\hat{Q}_{13}}^{\bar{B}_r}(\bar{A}_r-\bar{B}_rD^{-1}\bar{C}_r)+\hat{R}_3&=0,\label{v2.48}\\
      \bar{A}_r^T\Delta_{\hat{Q}_{22}}^{\bar{B}_r}+\Delta_{\hat{Q}_{22}}^{\bar{B}_r}\bar{A}_r+\hat{R}_4&=0,\label{v2.49}\\
      \bar{A}_r^T\Delta_{\hat{Q}_{23}}^{\bar{B}_r}+\Delta_{\hat{Q}_{23}}^{\bar{B}_r}(\bar{A}_r-\bar{B}_rD^{-1}\bar{C}_r)+\hat{R}_5&=0,\label{v2.50}\\
      (\bar{A}_r-\bar{B}_rD^{-1}\bar{C}_r)^T\Delta_{\hat{Q}_{33}}^{\bar{B}_r}+\Delta_{\hat{Q}_{33}}^{\bar{B}_r}(\bar{A}_r-\bar{B}_rD^{-1}\bar{C}_r)+\hat{R}_6&=0\label{v2.51}
      \end{align}
      wherein
      \begin{align}
      \hat{R}_1&=-C^TD^{-T}(\Delta_{\bar{B}_r})^T\hat{Q}_{13}^T-\hat{Q}_{13}\Delta_{\bar{B}_r}D^{-1}C-C^TD^{-T}\bar{B}_r^T(\Delta_{\hat{Q}_{13}}^{\bar{B}_r})^T-\Delta_{\hat{Q}_{13}}^{\bar{B}_r}\bar{B}_rD^{-1}C,\nonumber\\
      \hat{R}_2&=-C^TD^{-T}(\Delta_{\bar{B}_r})^T\hat{Q}_{23}+\hat{Q}_{13}\Delta_{\bar{B}_r}D^{-1}\bar{C}_r-C^TD^{-T}\bar{B}_r^T\Delta_{\hat{Q}_{23}}^{\bar{B}_r}+\Delta_{\hat{Q}_{13}}^{\bar{B}_r}\bar{B}_rD^{-1}\bar{C}_r,\nonumber\\
      \hat{R}_3&=-\hat{Q}_{13}\Delta_{\bar{B}_r}D^{-1}\bar{C}_r-C^TD^{-T}(\Delta_{\bar{B}_r})^T\hat{Q}_{33}-C^TD^{-T}\bar{B}_r^T\Delta_{\hat{Q}_{33}}^{\bar{B}_r},\nonumber\\
      \hat{R}_4&=\bar{C}_r^TD^{-T}(\Delta_{\bar{B}_r})^T\hat{Q}_{23}+\bar{C}_r^TD^{-T}\bar{B}_r^T\Delta_{\hat{Q}_{23}}^{\bar{B}_r}+\Delta_{\hat{Q}_{23}}^{\bar{B}_r}\bar{B}_rD^{-1}\bar{C}_r+\hat{Q}_{23}\Delta_{\bar{B}_r}D^{-1}\bar{C}_r,\nonumber\\
      \hat{R}_5&=-\hat{Q}_{23}\Delta_{\bar{B}_r}D^{-1}\bar{C}_r+\bar{C}_r^TD^{-T}(\Delta_{\bar{B}_r})^T\hat{Q}_{33}+\bar{C}_r^TD^{-T}\bar{B}_r^T\Delta_{\hat{Q}_{33}}^{\bar{B}_r},\nonumber\\
      \hat{R}_6&=-\bar{C}_r^TD^{-T}(\Delta_{\bar{B}_r})^T\hat{Q}_{33}-\hat{Q}_{33}\Delta_{\bar{B}_r}D^{-1}\bar{C}_r.\nonumber
      \end{align}
       By differentiating $\hat{J}$ with respect to $\bar{B}_r$, we get
      \begin{align}
      \Delta_{\hat{J}}^{\bar{B}_r}&=trace(B^T\Delta_{\hat{Q}_{11}}^{\bar{B}_r}B+2B^T\hat{Q}_{12}\Delta_{\bar{B}_r}+2B^T\Delta_{\hat{Q}_{12}}^{\bar{B}_r}\bar{B}_r+2\bar{B}_r^T\hat{Q}_{22}\Delta_{\bar{B}_r}+\bar{B}_r^T\Delta_{\hat{Q}_{22}}^{\bar{B}_r}\bar{B}_r)\nonumber\\
      &=trace(BB^T\Delta_{\hat{Q}_{11}}^{\bar{B}_r}+2\hat{Q}_{12}^TB(\Delta_{\bar{B}_r})^T+2B\bar{B}_r^T(\Delta_{\hat{Q}_{12}}^{\bar{B}_r})^T\nonumber\\
      &\hspace*{6.1cm}+2\hat{Q}_{22}\bar{B}_r(\Delta_{\bar{B}_r})^T+\bar{B}_r\bar{B}_r^T\Delta_{\hat{Q}_{22}}^{\bar{B}_r})\nonumber
      \end{align}
      By applying Lemma (\ref{lemma}) on the equations (\ref{v2.46}) and (\ref{nnee6}), (\ref{v2.47}) and (\ref{nnee7}), and (\ref{v2.49}) and (\ref{nnee8}), we get
      \begin{align}
      trace(BB^T\Delta_{\hat{Q}_{11}}^{\bar{B}_r})&=trace(\hat{R}_1P_{11}),\nonumber\\
      trace(B\bar{B}_r^T(\Delta_{\hat{Q}_{12}}^{\bar{B}_r})^T)&=trace(\hat{R}_2^TP_{12}),\nonumber\\
      trace(\bar{B}_r\bar{B}_r^T\Delta_{\hat{Q}_{22}}^{\bar{B}_r})&=trace(\hat{R}_4\hat{P}_{22}).\nonumber
      \end{align}
      Thus
      \begin{align}
      \Delta_{\hat{J}}^{\bar{B}_r}&=trace\Big(2\hat{Q}_{12}^TB(\Delta_{\bar{B}_r})^T+2\hat{Q}_{22}\bar{B}_r(\Delta_{\bar{B}_r})^T-2\hat{Q}_{13}^TP_{11}C^TD^{-T}(\Delta_{\bar{B}_r})^T\nonumber\\
      &\hspace*{2cm}+2\hat{Q}_{13}^TP_{12}\bar{C}_r^TD^{-T}(\Delta_{\bar{B}_r})^T-2\hat{Q}_{23}P_{12}^TC^TD^{-T}(\Delta_{\bar{B}_r})^T\nonumber\\
      &\hspace*{2cm}+2\hat{Q}_{23}\hat{P}_{22}\bar{C}_r^TD^{-T}(\Delta_{\bar{B}_r})^T-2P_{11}C^TD^{-T}\bar{B}_r^T(\Delta_{\hat{Q}_{13}}^{\bar{B}_r})^T\nonumber\\
      &\hspace*{2cm}+2P_{12}\bar{C}_r^TD^{-T}\bar{B}_r^T(\Delta_{\hat{Q}_{13}}^{\bar{B}_r})^T-2P_{12}^TC^TD^{-T}\bar{B}_r^T(\Delta_{\hat{Q}_{23}}^{\bar{B}_r})^T\nonumber\\
      &\hspace*{2cm}+2\hat{P}_{22}\bar{C}_r^TD^{-T}\bar{B}_r^T\Delta_{\hat{Q}_{23}}^{\bar{B}_r}\Big).\nonumber
      \end{align}
      By applying Lemma (\ref{lemma}) on the equations (\ref{v2.48}) and (\ref{v2.21}), and (\ref{v2.50}) and (\ref{v2.22}), we get
      \begin{align}
      trace(-P_{11}C^TD^{-T}\bar{B}_r^T(\Delta_{\hat{Q}_{13}}^{\bar{B}_r})^T+P_{12}\bar{C}_r^TD^{-T}\bar{B}_r^T(\Delta_{\hat{Q}_{13}}^{\bar{B}_r})^T)=trace(\hat{R}_{3}^T\hat{P}_{13}),\nonumber\\
      trace(-P_{12}^TC^TD^{-T}\bar{B}_r^T\Delta_{\hat{Q}_{23}}^{\bar{B}_r}+\hat{P}_{22}\bar{C}_r^TD^{-T}\bar{B}_r^T\Delta_{\hat{Q}_{23}}^{\bar{B}_r})=trace(\hat{R}_{5}^T\hat{P}_{23}).\nonumber
      \end{align}
      Thus
      \begin{align}
      \Delta_{\hat{J}}^{\bar{B}_r}&=trace\Big(2\hat{Q}_{12}^TB(\Delta_{\bar{B}_r})^T+2\hat{Q}_{22}\bar{B}_r(\Delta_{\bar{B}_r})^T-2\hat{Q}_{13}^TP_{11}C^TD^{-T}(\Delta_{\bar{B}_r})^T\nonumber\\
      &\hspace*{2cm}+2\hat{Q}_{13}^TP_{12}\bar{C}_r^TD^{-T}(\Delta_{\bar{B}_r})^T-2\hat{Q}_{23}P_{12}^TC^TD^{-T}(\Delta_{\bar{B}_r})^T\nonumber\\
      &\hspace*{2cm}+2\hat{Q}_{23}\hat{P}_{22}\bar{C}_r^TD^{-T}(\Delta_{\bar{B}_r})^T-2\hat{Q}_{13}^T\hat{P}_{13}\bar{C}_r^TD^{-T}(\Delta_{\bar{B}_r})^T\nonumber\\
      &\hspace*{2cm}-2\hat{Q}_{33}\hat{P}_{13}^TC^TD^{-T}(\Delta_{\bar{B}_r})^T-2\hat{Q}_{23}\hat{P}_{23}\bar{C}_r^TD^{-T}(\Delta_{\bar{B}_r})^T\nonumber\\
      &\hspace*{2cm}+2\hat{Q}_{33}\hat{P}_{23}^T\bar{C}_r^TD^{-T}(\Delta_{\bar{B}_r})^T-2\bar{B}_rD^{-1}C\hat{P}_{13}\Delta_{\hat{Q}_{33}}^{\bar{B}_r}\nonumber\\
      &\hspace*{2cm}+2\bar{B}_rD^{-1}\bar{C}_r\hat{P}_{23}\Delta_{\hat{Q}_{33}}^{\bar{B}_r}\Big).\nonumber
      \end{align}
      By applying Lemma \ref{lemma} on the equations (\ref{v2.51}) and (\ref{v2.23}), we get
      \begin{align}
      trace(-2\bar{B}_rD^{-1}C\hat{P}_{13}\Delta_{\hat{Q}_{33}}^{\bar{B}_r}+2\bar{B}_rD^{-1}\bar{C}_r\hat{P}_{23}\Delta_{\hat{Q}_{33}}^{\bar{B}_r})=trace(\hat{R}_6\hat{P}_{33}).\nonumber
      \end{align}
      Thus
      \begin{align}
      \Delta_{\hat{J}}^{\bar{B}_r}&=trace\Big(2\hat{Q}_{12}^TB+2\hat{Q}_{22}\bar{B}_r+\big(-2\hat{Q}_{13}^TP_{11}C^T+2\hat{Q}_{13}^TP_{12}\bar{C}_r^T-2\hat{Q}_{23}P_{12}^TC^T\nonumber\\
      &\hspace*{2cm}+2\hat{Q}_{23}\hat{P}_{22}\bar{C}_r^T-2\hat{Q}_{13}^T\hat{P}_{13}\bar{C}_r^T-2\hat{Q}_{33}\hat{P}_{13}^TC^T-2\hat{Q}_{23}\hat{P}_{23}\bar{C}_r^T\nonumber\\
      &\hspace*{2cm}+2\hat{Q}_{33}\hat{P}_{23}^T\bar{C}_r^T-2\hat{Q}_{33}\hat{P}_{33}\bar{C}_r^T\big)D^{-T}(\Delta_{\bar{B}_r})^T\Big)\nonumber
      \end{align}
      and
      \begin{align}
      \frac{\partial}{\partial\bar{B}_r}||\Delta_{mul}(s)||_{\mathcal{H}_2}^2=2(\hat{Q}_{12}^TB+\hat{Q}_{22}\bar{B}_r+\bar{Y}).\nonumber
      \end{align}
      Hence,
      \begin{align}
      \hat{Q}_{12}^TB+\hat{Q}_{22}\bar{B}_r+\bar{Y}=0\nonumber
      \end{align} is a necessary condition of for the local optimum of $||\Delta_{mul}||_{\mathcal{H}_2}^2$.

      Let us denote the first-order derivative of $\hat{J}$, $\hat{P}_{13}$, $\hat{P}_{23}$, and $\hat{P}_{33}$ with respect to $\bar{C}_r$ as $\Delta_{\hat{J}}^{\bar{C}_r}$, $\Delta_{\hat{P}_{13}}^{\bar{C}_r}$, $\Delta_{\hat{P}_{23}}^{\bar{C}_r}$, and $\Delta_{\hat{P}_{33}}^{\bar{C}_r}$, respectively. Further, let us denote the differential of $\bar{C}_r$ as $\Delta_{\bar{C}_r}$. By taking differentiation of the equations (\ref{v2.21})-(\ref{v2.23}) with respect to $\bar{C}_r$, we get
          \begin{align}
          A\Delta_{\hat{P}_{13}}^{\bar{C}_r}+\Delta_{\hat{P}_{13}}^{\bar{C}_r}(\bar{A}_r-\bar{B}_rD^{-1}\bar{C}_r)^T+\hat{T}_1&=0,\label{v2.52}\\
          \bar{A}_r\Delta_{\hat{P}_{23}}^{\bar{C}_r}+\Delta_{\hat{P}_{23}}^{\bar{C}_r}(\bar{A}_r-\bar{B}_rD^{-1}\bar{C}_r)^T+\hat{T}_2&=0,\label{v2.53}\\
          (\bar{A}_r-\bar{B}_rD^{-1}\bar{C}_r)\Delta_{\hat{P}_{33}}^{\bar{C}_r}+\Delta_{\hat{P}_{33}}^{\bar{C}_r}(\bar{A}_r-\bar{B}_rD^{-1}\bar{C}_r)^T+\hat{T}_3&=0.\label{v2.54}
          \end{align}
          wherein
          \begin{align}
          \hat{T}_1&=-\hat{P}_{13}(\Delta_{\bar{C}_r})^TD^{-T}\bar{B}_r^T+P_{12}(\Delta_{\bar{C}_r})^TD^{-T}\bar{B}_r^T,\nonumber\\
          \hat{T}_2&=\hat{P}_{23}(\Delta_{\bar{C}_r})^TD^{-T}\bar{B}_r^T+\hat{P}_{22}(\Delta_{\bar{C}_r})^TD^{-T}\bar{B}_r^T,\nonumber\\
          \hat{T}_3&=-\bar{B}_rD^{-1}\Delta_{\bar{C}_r}\hat{P}_{33}-\hat{P}_{33}(\Delta_{\bar{C}_r})^TD^{-T}\bar{B}_r^T-\bar{B}_rD^{-1}C\Delta_{\hat{P}_{13}}^{\bar{C}_r}\nonumber\\
          &\hspace*{2cm}+\bar{B}_rD^{-1}\Delta_{\bar{C}_r}\hat{P}_{23}+\bar{B}_rD^{-1}\bar{C}_r\Delta_{\hat{P}_{23}}^{\bar{C}_r}-(\Delta_{\hat{P}_{13}}^{\bar{C}_r})^TC^TD^{-T}\bar{B}_r^T\nonumber\\
          &\hspace*{2cm}+\hat{P}_{23}^T(\Delta_{\bar{C}_r})^TD^{-T}\bar{B}_r^T+(\Delta_{\hat{P}_{23}}^{\bar{C}_r})^T\bar{C}_r^TD^{-T}\bar{B}_r^T.\nonumber
          \end{align}
          By taking differentiation of $\hat{J}$ with respect to $\bar{C}_r$, we get
          \begin{align}
      \Delta_{\hat{J}}^{\bar{C}_r}&=trace\Big(-2D^{-1}CP_{12}(\Delta_{\bar{C}_r})^TD^{-T}+2D^{-1}C\hat{P}_{13}(\Delta_{\bar{C}_r})^TD^{-T}\nonumber\\
      &\hspace*{2cm}+2D^{-1}\bar{C}_r\hat{P}_{22}(\Delta_{\bar{C}_r})^TD^{-T}-2D^{-1}\bar{C}_r\hat{P}_{23}(\Delta_{\bar{C}_r})^TD^{-T}\nonumber\\
      &\hspace*{2cm}-2D^{-1}\bar{C}_r\hat{P}_{23}^T(\Delta_{\bar{C}_r})^TD^{-T}+2D^{-1}\bar{C}_r\hat{P}_{33}(\Delta_{\bar{C}_r})^TD^{-T}\nonumber\\
      &\hspace*{2cm}+2D^{-1}C\Delta_{\hat{P}_{13}}^{\bar{C}_r}\bar{C}_r^TD^{-T}-2D^{-1}\bar{C}_r\Delta_{\hat{P}_{23}}^{\bar{C}_r}\bar{C}_r^TD^{-T}\nonumber\\
      &\hspace*{2cm}+D^{-1}\bar{C}_r\Delta_{\hat{P}_{33}}^{\bar{C}_r}\bar{C}_r^TD^{-T}\Big)\nonumber\\
      &=trace\Big(-2D^{-T}D^{-1}CP_{12}(\Delta_{\bar{C}_r})^T+2D^{-T}D^{-1}C\hat{P}_{13}(\Delta_{\bar{C}_r})^T\nonumber\\
      &\hspace*{2cm}+2D^{-T}D^{-1}\bar{C}_r\hat{P}_{22}(\Delta_{\bar{C}_r})^T-2D^{-T}D^{-1}\bar{C}_r\hat{P}_{23}(\Delta_{\bar{C}_r})^T\nonumber\\
      &\hspace*{2cm}-2D^{-T}D^{-1}\bar{C}_r\hat{P}_{23}^T(\Delta_{\bar{C}_r})^T+2D^{-T}D^{-1}\bar{C}_r\hat{P}_{33}(\Delta_{\bar{C}_r})^T\nonumber\\
      &\hspace*{2cm}+2C^TD^{-T}D^{-1}\bar{C}_r(\Delta_{\hat{P}_{13}}^{\bar{C}_r})^T-2\bar{C}_r^TD^{-T}D^{-1}\bar{C}_r(\Delta_{\hat{P}_{23}}^{\bar{C}_r})^T\nonumber\\
      &\hspace*{2cm}+\bar{C}_r^TD^{-T}D^{-1}\bar{C}_r\Delta_{\hat{P}_{33}}^{\bar{C}_r}\Big)\nonumber
      \end{align}
      By applying Lemma (\ref{lemma}) on the equations (\ref{v2.52}) and (\ref{v2.27}), (\ref{v2.53}) and (\ref{v2.29}), and (\ref{v2.54}) and (\ref{v2.30}), we get
      \begin{align}
      trace(\hat{T}_1^T\hat{Q}_{13})&=trace(-C^TD^{-T}\bar{B}_r^T\hat{Q}_{33}(\Delta_{\hat{P}_{13}}^{\bar{C}_r})^T+C^TD^{-T}D^{-1}\bar{C}_r(\Delta_{\hat{P}_{13}}^{\bar{C}_r})^T),\nonumber\\
      trace(\hat{T}_2^T\hat{Q}_{23})&=trace(\bar{C}_r^TD^{-T}\bar{B}_r^T\hat{Q}_{33}(\Delta_{\hat{P}_{23}}^{\bar{C}_r})^T-\bar{C}_r^TD^{-T}D^{-1}\bar{C}_r(\Delta_{\hat{P}_{23}}^{\bar{C}_r})^T),\nonumber\\
      trace(\hat{T}_3\hat{Q}_{33})&=trace(\bar{C}_r^TD^{-T}D^{-1}\bar{C}_r\Delta_{\hat{P}_{33}}^{\bar{C}_r}).\nonumber
      \end{align}
      Thus
      \begin{align}
      \Delta_{\hat{J}}^{\bar{C}_r}&=trace\Big(D^{-T}\big(-2D^{-1}CP_{12}+2D^{-1}C\hat{P}_{13}+2D^{-1}\bar{C}_r\hat{P}_{22}-2D^{-1}\bar{C}_r\hat{P}_{23}\nonumber\\
      &\hspace*{1.5cm}-2D^{-1}\bar{C}_r\hat{P}_{23}^T+2D^{-1}\bar{C}_r\hat{P}_{33}+2\bar{B}_r^T\hat{Q}_{13}^T\hat{P}_{13}+2\bar{B}_r^T\hat{Q}_{13}^TP_{12}\nonumber\\
      &\hspace*{1.5cm}+2\bar{B}_r^T\hat{Q}_{23}\hat{P}_{23}+2\bar{B}_r^T\hat{Q}_{23}P_{22}+2\bar{B}^T\hat{Q}_{33}\hat{P}_{33}+2\bar{B}_r^T\hat{Q}_{33}\hat{P}_{23}^T\big)(\Delta_{\bar{C}_r})^T\Big)\nonumber
      \end{align}
      Hence,
      \begin{align}
      \frac{\partial}{\partial\bar{C}_r}||\Delta_{mul}(s)||_{\mathcal{H}_2}^2=2D^{-T}(-D^{-1}CP_{12}+D^{-1}C\hat{P}_{13}+D^{-1}\bar{C}_r\hat{P}_{22}-D^{-1}\bar{C}_r\hat{P}_{23}\nonumber\\
      \hspace*{1.5cm}-D^{-1}\bar{C}_r\hat{P}_{23}^T+D^{-1}\bar{C}_r\hat{P}_{33}+\bar{B}_r^T\hat{Q}_{13}^T\hat{P}_{13}+\bar{B}_r^T\hat{Q}_{13}^TP_{12}\nonumber\\
      \hspace*{1.5cm}+\bar{B}_r^T\hat{Q}_{23}\hat{P}_{23}+\bar{B}_r^T\hat{Q}_{23}\hat{P}_{22}+\bar{B}_r^T\hat{Q}_{33}\hat{P}_{33}+\bar{B}_r^T\hat{Q}_{33}\hat{P}_{23}^T)\nonumber
      \end{align} and
    \begin{align}
     -(D^{T}D)^{-1}CP_{12}+(D^{T}D)^{-1}\bar{C}_r\hat{P}_{22}+\bar{Z}&=0\nonumber
    \end{align}is a necessary condition for the local optimum of $||\Delta_{mul}||_{\mathcal{H}_2}^2$.
\section*{Acknowledgment}
This work is supported by the National Natural Science Foundation of China under Grant No. 61873336, the International Corporation Project of Shanghai Science and Technology Commission under Grant 21190780300, in part by The National Key Research and Development Program No. 2020YFB1708200 , and in part by the High-end foreign expert programs No. QN2022013008L, G2021013008L, and G2022013050L granted by the State Administration of Foreign Experts Affairs (SAFEA).

\end{document}